\documentclass[11pt,finalcls,onecolumn]{IEEEtran}
\usepackage{graphicx}
\usepackage{epstopdf}
\usepackage{epsfig}
\usepackage{amsmath}
\usepackage{amssymb}
\usepackage{float}
\usepackage{url}

\usepackage[bottom=0.7in]{geometry}

\newtheorem{corollary}{Corollary}
\newtheorem{definition}{Definition}
\newtheorem{fact}{Fact}
\newtheorem{notation}{Notation}
\newtheorem{theorem}{Theorem}
\newtheorem{proposition}{Proposition}
\newtheorem{lemma}{Lemma}
\newtheorem{remark}{Remark}
\newtheorem{example}{Example}
\newtheorem{question}{Question}

\DeclareMathOperator*{\esssup}{ess\,sup}

\newcommand{\naturals}{\ensuremath{\mathbb{N}}}
\newcommand{\Reals}{\ensuremath{\mathbb{R}}}
\newcommand{\probability}{\ensuremath{\mathbb{P}}}
\newcommand{\overbar}[1]{\mkern 1.5mu\overline{\mkern-1.5mu#1\mkern-1.5mu}\mkern 1.5mu}

\begin{document}

\title{On R\'{e}nyi Entropy Power Inequalities}
\author{\vspace*{0.8cm} Eshed Ram \qquad Igal Sason
\thanks{
E. Ram and I. Sason are with the Andrew and Erna Viterbi
Faculty of Electrical Engineering, Technion--Israel
Institute of Technology, Haifa 32000, Israel. E-mails:
\{s6eshedr@tx, sason@ee\}.technion.ac.il.}
\thanks{
This work has been supported by the Israeli Science Foundation
(ISF) under Grant 12/12. It has been presented in part at the
2016 IEEE International Symposium on Information Theory, Barcelona,
Spain, July~10--15, 2016.}}

\maketitle

\begin{abstract}
This paper gives improved R\'{e}nyi entropy power inequalities (R-EPIs).
Consider a sum $S_n = \sum_{k=1}^n X_k$ of $n$ independent continuous random
vectors taking values on $\Reals^d$, and let $\alpha \in [1, \infty]$. An R-EPI
provides a lower bound on the order-$\alpha$ R\'enyi entropy power of $S_n$ that,
up to a multiplicative constant (which may depend in general on $n, \alpha, d$),
is equal to the sum of the order-$\alpha$ R\'enyi entropy powers of the $n$ random
vectors $\{X_k\}_{k=1}^n$. For $\alpha=1$, the R-EPI coincides with the well-known
entropy power inequality by Shannon. The first improved R-EPI is obtained by
tightening the recent R-EPI by Bobkov and Chistyakov which relies on the sharpened
Young's inequality. A further improvement of the R-EPI also relies on convex
optimization and results on rank-one modification of a real-valued diagonal matrix.
\end{abstract}


\section{Introduction}
\label{section: Introduction}

One of the well-known inequalities in information theory is
the entropy power inequality (EPI) which has
been introduced by Shannon \cite[Theorem~15]{Shannon}. Let $X$ be a
$d$-dimensional random vector with a probability density function,
let $h(X)$ be its differential entropy, and let
$N(X) = \exp\left(\tfrac2d \, h(X)\right)$
be the entropy power of $X$. The EPI states that for
independent random vectors $\{X_k\}_{k=1}^n$, the following
inequality holds:
\begin{align} \label{eq: EPI}
N\left( \sum_{k=1}^n X_k \right) \geq \sum_{k=1}^n N(X_k)
\end{align}
with equality in \eqref{eq: EPI} if and only if $\{X_k\}_{k=1}^n$
are Gaussian random vectors with proportional covariances.

The EPI has proved to be an instrumental tool in proving converse theorems
for the capacity region of the Gaussian broadcast channel \cite{Bergmans},
the Gaussian wire-tap channel \cite{CheongH}, the capacity region
of the Gaussian broadcast multiple-input multiple-output (MIMO)
channel \cite{WeingartenSS}, and a converse theorem in multi-terminal
lossy compression \cite{Oohama}. Due to its importance, the EPI
has been proved with information-theoretic tools in several insightful ways
(see, e.g., \cite{Blachman}, \cite{DemboCoverThomas}, \cite{GouSV-ISIT06},
\cite[Appendix~D]{Johnson_book}, \cite{Rioul11}, \cite{Stam}, \cite{VerduG06}); e.g.,
the proof in \cite{VerduG06} relies on fundamental relations between
information and estimation measures (\cite{GSV05}, \cite{GSV12}), together with
the simple fact that for estimating a sum of two random variables, it is preferable to
have access to the individual noisy measurements rather than to their sum.
More studies on the theme include EPIs for discrete random variables and some
analogies \cite{SaeidAbbeTeleatar, HarremoesV03, JogAnanatharam, JohnsonYu,
ShamaiWyner, SharmaDM-ISIT11, WooMadiman15}, generalized EPIs
\cite{LiuV07, MadimanBarron07, ZamirF93}, reverse EPIs
\cite{BobkovMadiman11, BobkovMadiman13, MadimanMelbourne, Xu16},
related inequalities to the EPI in terms of rearrangements \cite{WangM-IT14},
and some refined versions of the EPI for specialized distributions
\cite{Costa, Courtade16, HarremoesV03, Toscani}. An overview on EPIs is provided in
\cite{Anantharam}; we also refer the reader to a preprint of a recent survey
paper by Madiman {\em et al.} \cite{MadimanMelbourne} which addresses forward
and reverse EPIs with R\'enyi measures, and their connections with convex geometry.

\vspace*{0.1cm}
The R\'{e}nyi entropy and divergence have been introduced in \cite{Renyi},
and they evidence a long track record of usefulness in information theory
and its applications. Recent studies of the properties of these R\'{e}nyi
measures have been provided in \cite{ErvenHarremoes}, \cite{FehrB14} and
\cite{Shayevitz}. In the following, the differential R\'{e}nyi entropy
and the R\'{e}nyi entropy power are introduced.

\vspace*{0.1cm}
\begin{definition}[Differential R\'{e}nyi entropy]
\label{definition: Renyi Entropy}
Let $X$ be a random vector which takes values in $\Reals^d$,
and assume that it has a probability density function which is designated
by $f_X$. The differential R\'{e}nyi entropy of $X$ of order
$\alpha \in (0,1) \cup (1, \infty)$, denoted by $h_{\alpha}(X)$,
is given by
\begin{align}
h_{\alpha}(X) &= \frac1{1-\alpha}
\; \log \Biggl( \; \int\limits_{\Reals^d} f_X^\alpha(x)
\, \mathrm{d}x \Biggr) \label{eq0: Renyi entropy} \\
\label{eq: Renyi entropy}
&= \frac{\alpha}{1-\alpha} \; \log \|f_X\|_{\alpha}.
\end{align}
The differential R\'{e}nyi entropies of orders $\alpha = 0,1,\infty$
are defined by the continuous extension of $h_{\alpha}(X)$ for
$\alpha \in (0,1) \cup (1, \infty)$, which yields
\begin{align}
\label{RE at zero}
& h_0(X) = \log \, \lambda\bigl(\mathrm{supp}(f_X)\bigr),  \\
\label{RE at 1}
& h_1(X) = h(X) = -\int\limits_{\Reals^d} f_X(x) \, \log f_X(x)
\, \mathrm{d}x, \\
\label{RE at infinity}
& h_{\infty}(X) = -\log \bigl( \esssup(f_X) \bigr)
\end{align}
where $\lambda$ in \eqref{RE at zero} is the
Lebesgue measure in $\Reals^d$.
\end{definition}

\begin{definition}[R\'{e}nyi entropy power]
\label{definition: Renyi entropy power}
For a $d$-dimensional random vector $X$ with density, the R\'{e}nyi entropy
power of order $\alpha \in [0, \infty]$ is given by
\begin{align} \label{eq: Renyi entropy power}
N_{\alpha}(X) = \exp\left(\tfrac{2}{d} \, h_{\alpha}(X) \right).
\end{align}
\end{definition}

Since $h_{\alpha}(X)$ is specialized to the Shannon entropy $h(X)$
for $\alpha=1$, the possibility of generalizing the EPI with R\'{e}nyi entropy
powers has emerged. This question is stated as follows:

\vspace*{0.1cm}
\begin{question} \label{question: R-EPI}
Let $\{X_k\}$ be independent $d$-dimensional random vectors with probability
density functions, and let $\alpha \in [0, \infty]$ and $n \in \naturals$.
Does a R\'enyi entropy power inequality (R-EPI) of the form
\begin{align} \label{Intro: R-EPI}
N_{\alpha}\left(\sum_{k=1}^n X_k \right) \geq c_{\alpha}^{(n,d)} \,
\sum_{k=1}^n N_{\alpha}(X_k)
\end{align}
hold for some positive constant $c_{\alpha}^{(n,d)}$ (which may depend on the
order $\alpha$, dimension $d$, and number of summands $n$) ?
\end{question}

\vspace*{0.1cm}
In \cite[Theorem~2.4]{JohnsonV07}, a sort of an R-EPI for the R\'{e}nyi
entropy of order $\alpha \geq 1$ has been derived with some analogy to
the classical EPI; this inequality, however, does not apply the usual
convolution unless $\alpha=1$. In \cite[Conjectures~4.3, 4.4]{WangM-IT14}, Wang
and Madiman conjectured an R-EPI for an arbitrary finite number of
independent random vectors in $\Reals^d$ for $\alpha > \frac{d}{d+2}$.

Question~\ref{question: R-EPI} has been recently addressed by
Bobkov and Chistyakov \cite{BobkovC15}, showing that \eqref{Intro: R-EPI}
holds with
\begin{align} \label{eq: c for BC}
c_{\alpha} = \tfrac1{e} \, \alpha^{\frac1{\alpha-1}}, \quad \forall \, \alpha > 1
\end{align}
independently of the values of $n, d$. It is the purpose of this paper to
derive some improved R-EPIs for $\alpha>1$ (the case of $\alpha=1$
refers to the EPI \eqref{eq: EPI}). A study of Question~\ref{question: R-EPI} for
$\alpha \in (0, 1)$ is currently an open problem (see \cite[p.~709]{BobkovC15}).

In view of the close relation in \eqref{eq: Renyi entropy} between the
(differential) R\'{e}nyi entropy and the $L_\alpha$ norm, the sharpened
version of Young's inequality plays a key role in \cite{BobkovC15} for
the derivation of an R-EPI, as well as in our paper for the derivation of
some improved R-EPIs. The sharpened version of Young's inequality was also
used by Dembo {\em et al.} \cite{DemboCoverThomas} for proving the EPI.

For $\alpha \in (1, \infty)$, let $\alpha' = \frac{\alpha}{\alpha-1}$ be
H\"older's conjugate. For $\alpha > 1$, Theorem~\ref{theorem: REPI1} provides
a new tighter constant in comparison to \eqref{eq: c for BC} which gets the form
\begin{align} \label{eq: R-EPI1}
c_\alpha^{(n)} = \alpha^{\frac1{\alpha-1}} \left(1-\frac1{n\alpha'}\right)^{n\alpha'-1}
\end{align}
independently of the dimension $d$. The new R-EPI with the constant in
\eqref{eq: R-EPI1} asymptotically coincides with the tight bound by Rogozin
\cite{Rogozin} when $\alpha \to \infty$ and $n=2$, and it also asymptotically
coincides with the R-EPI in \cite{BobkovC15} when $n \to \infty$.
Moreover, the R-EPI with the new constant in \eqref{eq: R-EPI1} is further
improved in Theorem~\ref{theorem: tightest REPI} by a more involved analysis
which relies on convex analysis and some interesting results from matrix theory;
the latter result yields a closed-form solution for $n=2$.

This paper is organized as follows:
In Section~\ref{section: Preliminaries}, preliminary material and notation
are introduced. A new R-EPI is derived in Section~\ref{section: A New Renyi EPI}
for $\alpha > 1$, and special cases of this improved bound are studied.
Section~\ref{section: further tightening of R-EPI} derives a strengthened R-EPI
for a sum of $n \geq 2$ random variables; for $n=2$, it is specialized to a bound
which is expressed in a closed form; its computation for $n>2$ requires a numerical
optimization which is easy to perform. Section~\ref{section: example} exemplifies
numerically the tightness of the new R-EPIs in comparison to some previously reported
bounds, and finally Section~\ref{section: summary} summarizes the paper.

\section{Analytical Tools}
\label{section: Preliminaries}
This section includes notation and tools which are essential to
the analysis in this paper. It starts with the sharpened Young's
inequality, followed by results on rank-one modification of a
symmetric eigenproblem \cite{Bunch}. We also include here some
properties of the differential R\'{e}nyi entropy and R\'{e}nyi
entropy power which are useful to the analysis in this paper.

\subsection{Basic Inequalities}
\label{subsection:Pre - Holder's and Young's Inequalities}
The derivation of the R-EPIs in this work partially relies on the
sharpened Young's inequality and the monotonicity of the R\'enyi
entropy in its order. For completeness, we introduce these results
in the following.
\begin{notation} \label{notation:Pre - holder conjugate}
For $\alpha>0$, let $\alpha' = \frac{\alpha}{\alpha-1}$, i.e.,
$\frac1{\alpha}+ \frac1{\alpha'}=1.$
\end{notation}
Note that $\alpha>1$ if and only if $\alpha'>0$; if $\alpha=1$, we define
$\alpha'=\infty$. This notation is known as H\"{o}lder's conjugate.
\begin{fact}[Monotonicity of the R\'enyi entropy] \label{fact:Pre - Holder inequality}
The R\'enyi entropy, $h_{\alpha}(X)$, is monotonically non-increasing in $\alpha$.
\end{fact}

From \eqref{eq: Renyi entropy}, it follows that for $\alpha \in (0,1) \cup (1, \infty)$,
if $f$ is a probability density function of a $d$-dimensional vector $X$, then
\begin{align}
h_{\alpha}(X) = -\log \bigl( \|f\|_{\alpha}^{\alpha'} \bigr). \label{eq: RE2}
\end{align}
A useful consequence of Fact~\ref{fact:Pre - Holder inequality} and \eqref{eq: RE2} is the following
result (a weaker version of it is given in \cite[Lemma~1]{BobkovC15}):
\begin{corollary}\label{corollary:Pre - Holder for densities}
Let $\alpha \in (0,1) \cup (1, \infty)$, and let $f \in L^\alpha(\Reals^d)$ be a probability
density function (i.e., $f$ is a non-negative function with $\|f\|_1=1$).
Then, for every $\beta \in (0, \alpha)$ with $\beta \neq 1$,
\begin{align} \label{inequality:Pre -  Holder for densities}
\|f\|_{\beta}^{\beta'} \leq  \|f\|_{\alpha}^{\alpha'}.
\end{align}
\end{corollary}

\begin{notation} \label{notation:Pre - At}
For every $t \in (0,1) \cup (1, \infty)$, let
\begin{align} \label{eq:Pre - At}
A_t = t^{\frac1{t}} \, |t'|^{-\frac1{|t'|}}
\end{align}
and let $A_1 = A_{\infty} = 1$. Note that for $t \in [0, \infty]$
\begin{align} \label{eq: property A}
A_{t'} = \frac1{A_t}.
\end{align}
\end{notation}
The sharpened Young's inequality, first derived
by Beckner \cite{Beckner} and re-derived with
alternative proofs in, e.g., \cite{Barthe} and
\cite{BrascampLieb} is given as follows:
\begin{fact}[Sharpened Young's inequality]
\label{fact:Pre - Sharpened Young's inequality}
Let $p,q,r>0$ satisfy
\begin{align} \label{eq: Young's condition}
\frac1{p}+ \frac1{q} = 1 + \frac1{r},
\end{align}
let $f \in L^p(\Reals^d)$ and $g \in L^q(\Reals^d)$
be non-negative functions, and let $f \ast g$ denote
their convolution.
\begin{itemize}
\item If $p,q,r>1$, then
\begin{align} \label{eq: Young's ineq. 1}
\|f \ast g\|_r \leq \left( \frac{A_p A_q}{A_r}\right)^{\frac{d}{2}}
\|f\|_p \, \|g\|_q.
\end{align}
\item If $p,q,r<1$, then
\begin{align} \label{eq: Young's ineq. 2}
\|f \ast g\|_r \geq \left( \frac{A_p A_q}{A_r}\right)^{\frac{d}{2}}
\|f\|_p \, \|g\|_q.
\end{align}
\end{itemize}
Furthermore, \eqref{eq: Young's ineq. 1} and \eqref{eq: Young's ineq. 2}
hold with equalities if and only if $f$ and $g$ are Gaussian probability
density functions.
\end{fact}

Note that the condition in \eqref{eq: Young's condition} can be expressed
in terms of the H\"older's conjugates as follows:
\begin{align} \label{alternative form}
\frac1{p'}+ \frac1{q'} = \frac1{r'}.
\end{align}
By using \eqref{alternative form} and mathematical induction,
the sharpened Young's inequality can be extended to more than
two functions as follows:
\begin{corollary}\label{corollary: Young's ineq.}
Let $\nu, \{\nu_k\}_{k=1}^n > 0$ satisfy
$\sum_{k=1}^n\frac1{\nu_k'}=\frac1{\nu'}$,
let
\begin{align} \label{eq:Pre - Multiple Beckner's Best Constant}
A = \left(\frac1{A_\nu} \prod_{k=1}^n A_{\nu_k}\right)^{\frac{d}{2}}
\end{align}
where the right side in \eqref{eq:Pre - Multiple Beckner's Best Constant}
is defined by \eqref{eq:Pre - At}, and let $f_k \in L^{\nu_k}(\Reals^d)$
be non-negative functions.
\begin{itemize}
\item If $\nu,\{\nu_k\}_{k=1}^n>1$, then
\begin{align} \label{eq: Young's ineq. 3}
\|f_1 \ast \ldots \ast f_n\|_\nu \leq A \prod_{k=1}^n \|f_k\|_{\nu_k} .
\end{align}
\item If $\nu,\{\nu_k\}_{k=1}^n<1$, then
\begin{align} \label{eq: Young's ineq. 4}
\|f_1 \ast \ldots \ast f_n\|_\nu \geq A \prod_{k=1}^n \|f_k\|_{\nu_k}
\end{align}
\end{itemize}
with equalities in \eqref{eq: Young's ineq. 3} and
\eqref{eq: Young's ineq. 4} if and only if $f_k$ are
scaled versions of Gaussian probability densities for all $k$.
\end{corollary}

\subsection{Rank-One Modification of a Symmetric Eigenproblem}
\label{subsection:Pre - Rank-One Modification}
This section is based on a paper by Bunch {\em et al.} \cite{Bunch}
which addresses the eigenvectors and eigenvalues (a.k.a.
eigensystem) of rank-one modification of a real-valued
diagonal matrix.
We use in this paper the following result \cite{Bunch}:
\begin{fact} \label{fact: Bunch}
Let $D \in \Reals^{n \times n}$ be a diagonal matrix with the
eigenvalues $d_1 \leq d_2 \leq \ldots \leq d_n $.
Let $z \in \Reals^n$ such that $\|z\|_2 = 1$ and let
$\rho \in \Reals$. Let
$\lambda_1 \leq \lambda_2 \leq \ldots \leq \lambda_n$
be the eigenvalues of the rank-one modification of $D$
which is given by $C=D+\rho z z^T$. Then,
\begin{enumerate}
\item \label{item1: Bunch}
$\lambda_i = d_i + \rho \mu_i$, where $\sum_{i=1}^n \mu_i = 1$
and $\mu_i \geq 0$ for all $i \in \{1, \ldots, n\}$.
\item \label{item2: Bunch}
If $\rho > 0$, then the following interlacing property holds:
\begin{align} \label{eq: Bunch1}
d_1 \leq \lambda_1 \leq d_2 \leq \lambda_2 \leq \ldots \leq d_n \leq \lambda_n
\end{align}
and, if $\rho < 0$, then
\begin{align} \label{eq: Bunch2}
\lambda_1 \leq  d_1 \leq \lambda_2 \leq d_2 \leq \ldots \leq \lambda_n \leq d_n.
\end{align}
\item \label{item3: Bunch}
If all the eigenvalues of $D$ are different, all the entries of $z$ are non-zero,
and $\rho \neq 0$, then inequalities \eqref{eq: Bunch1} and \eqref{eq: Bunch2} are
strict. For $i \in \{1, \ldots, n\}$, the eigenvalue $\lambda_i$ is a zero of
\begin{align} \label{eq:Pre - Rank1 Theorem, W function }
W(x) = 1+\rho \sum_{j=1}^n \frac{z_i^2}{d_j-x}.
\end{align}
\end{enumerate}
\end{fact}
Note that the requirement $\|z\|_2 = 1$ can be relaxed to $z \neq 0$
by letting $\hat{z}= \frac{z}{\|z\|_2}$ and $\hat{\rho}= \rho \|z\|_2^2$.

\subsection{R\'{e}nyi Entropy Power}
\label{subsection:Pre - Renyi entropy}
We present some properties of the differential R\'{e}nyi entropy and
R\'{e}nyi entropy power which are useful in this paper.
\begin{itemize}
\item In view of \eqref{eq: Renyi entropy} and \eqref{eq: Renyi entropy power},
for $\alpha \in (0,1) \cup (1, \infty)$,
\begin{align} \label{eq:Pre - Renyi entropy power Norm dfn}
N_{\alpha}(X)= \left( \|f_X\|_\alpha \right)^{-\frac{2\alpha'}{d}}.
\end{align}
\item
The differential R\'{e}nyi entropy $h_{\alpha}(X)$ is monotonically
non-increasing in $\alpha$, and so is $N_\alpha(X)$.
\item
If $Y=AX+b$ where $A \in \Reals^{d \times d}, |A| \neq 0$, $b \in  \Reals^{d}$,
then for all $\alpha \in [0,\infty]$
\begin{align}
& h_\alpha(Y) = h_\alpha(X) + \log|A|, \\
& N_\alpha(Y) = |A|^{\frac2d} \, N_{\alpha}(X).
\end{align}
This implies that the R\'{e}nyi entropy power is a homogeneous functional
of order~2 and it is translation invariant, i.e.,
\begin{align} \label{eq: REP}
& N_{\alpha}(\lambda X) =\lambda^2 \, N_{\alpha}(X),
\quad \forall \, \lambda \in \Reals,\\
\label{eq2: REP}
& N_{\alpha}(X+b) = N_{\alpha}(X), \quad \forall \, b \in \Reals^d.
\end{align}
\end{itemize}
In view of \eqref{eq: REP} and \eqref{eq2: REP}, $N_{\alpha}(X)$ has some
similar properties to the variance of $X$. However, if we
consider a sum of independent random vectors then
$\text{Var} \left( \sum_{k=1}^n X_k\right)=\sum_{k=1}^n \text{Var}(X_k)$
whereas the R\'{e}nyi entropy power of a sum of independent random vectors
is not equal, in general, to the sum of the R\'{e}nyi entropy powers of the
individual random vectors (unless these independent vectors are Gaussian
with proportional covariances).

The continuation of this paper considers R-EPIs for orders $\alpha \in (1, \infty]$.
The case where $\alpha=1$ refers to the EPI by Shannon \cite[Theorem~15]{Shannon}.

\section{A New R\'{e}nyi EPI}
\label{section: A New Renyi EPI}
In the following, a new R-EPI is derived. This inequality, which is expressed
in closed-form, is tighter than the R-EPI in \cite[Theorem~I.1]{BobkovC15}.
\begin{theorem} \label{theorem: REPI1}
Let $\{X_k\}_{k=1}^n$ be independent random vectors with densities defined on
$\Reals^d$, and let $n \in \naturals$, $\alpha >1$, $\alpha' = \frac{\alpha}{\alpha-1}$
and $S_n = \sum_{k=1}^n X_k$. Then, the following R-EPI holds:
\begin{align} \label{eq: REPI1}
N_{\alpha}(S_n) \geq c_\alpha^{(n)} \sum_{k=1}^n N_{\alpha}(X_k)
\end{align}
with
\begin{align} \label{eq: c for R-EPI1}
c_\alpha^{(n)} =\alpha^{\frac1{\alpha-1}} \left(1-\frac1{n\alpha'} \right)^{n\alpha'-1}.
\end{align}
Furthermore, the R-EPI in \eqref{eq: REPI1} has the following properties:
\begin{enumerate}
\item Eq.~\eqref{eq: REPI1} improves the R-EPI in \cite[Theorem~I.1]{BobkovC15} for
every $\alpha > 1$ and $n \in \naturals$,
\item For all $\alpha > 1$, it asymptotically coincides with the R-EPI in
\cite[Theorem~I.1]{BobkovC15} as $n \to \infty$,
\item In the other limiting case where $\alpha \downarrow 1$, it coincides with the
EPI (similarly to \cite{BobkovC15}),
\item If $n=2$ and $\alpha \to \infty$, the constant $c_\alpha^{(n)}$ in
\eqref{eq: c for R-EPI1} tends to $\tfrac12$ which is optimal;
this constant is achieved when $X_1$ and $X_2$ are independent random vectors
which are uniformly distributed in the cube $[0,1]^d$.
\end{enumerate}
\end{theorem}

\begin{proof}
In the first stage of this proof, we assume that
\begin{align} \label{eq: positive N_k}
N_\alpha(X_k) > 0, \quad k \in \{1, \ldots, n\}
\end{align}
which, in view of \eqref{eq:Pre - Renyi entropy power Norm dfn},
implies that $f_{X_k} \in L^\alpha(\Reals^d)$, where $f_{X_k}$
is the density of $X_k$ for all $k \in \{1,\ldots,n\}$.
In \cite[(12)]{BobkovC15} it is shown that for $\alpha>1$,
\begin{align} \label{eq: R-EPI BC15}
N_{\alpha}(S_n) \geq B \prod_{k=1}^n N_{\alpha}^{t_k}(X_k)
\end{align}
with
\begin{align} \label{eq1: BC1}
& B = \bigl( A_{\nu_1} \ldots A_{\nu_n} A_{\alpha'} \bigr)^{-\alpha'}, \\
\label{eq2: BC1}
& \nu_k > 1, \quad \forall \, k \in \{1, \ldots, n\}, \\
\label{eq3: BC1}
& \nu' = \frac{\nu}{\nu-1}, \quad \forall \, \nu \in \Reals, \\
\label{eq4: BC1}
& \sum_{k=1}^n \frac1{\nu'_k} = \frac1{\alpha'}, \\
\label{eq5: BC1}
& t_k = \frac{\alpha'}{\nu'_k}, \quad \forall \, k \in \{1, \ldots, n\}.
\end{align}
Consequently, \eqref{eq2: BC1}--\eqref{eq5: BC1} yields
\begin{align}  \label{eq1: BC2}
& t_k > 0, \quad \forall \, k \in \{1, \ldots, n\}, \\
\label{eq2: BC2}
& \sum_{k=1}^n t_k = 1.
\end{align}
The proof of \eqref{eq: R-EPI BC15},
which relies on Corollaries~\ref{corollary:Pre - Holder for densities}
and~\ref{corollary: Young's ineq.}, is introduced in Appendix~\ref{appendix: BC15}.

Similarly to \cite[(14)]{BobkovC15}, in view of the homogeneity
of the entropy power functional (see \eqref{eq: REP}), it can be
assumed without any loss of generality that
\begin{align} \label{eq: normalized REP}
\sum_{k=1}^n N_{\alpha}(X_k)=1.
\end{align}
Hence, to prove \eqref{eq: REPI1}, it is sufficient to show that
under the assumption in \eqref{eq: normalized REP}
\begin{align} \label{eq: equiv. REPI1}
N_{\alpha}(S_n) \geq c_\alpha^{(n)}.
\end{align}

From this point, we deviate from the proof of
\cite[Theorem~I.1]{BobkovC15}. Taking logarithms
on both sides of \eqref{eq: R-EPI BC15} and assembling
\eqref{eq:Pre - At}, \eqref{eq1: BC1}--\eqref{eq2: BC2}
and \eqref{eq: normalized REP} yield
\begin{align} \label{key inequality}
\log N_{\alpha}(S_n) \geq  f_0(\underline{t}),
\end{align}
where $\underline{t}=(t_1,\ldots,t_n )$, and
\begin{align} \label{eq: f0}
& f_0(\underline{t}) = \frac{\log \alpha}{\alpha-1}-D(\underline{t}\|
\underline{N}_\alpha) + \alpha' \sum_{k=1}^n\left( 1-\frac{t_k}{\alpha'}
\right) \log \left( 1-\frac{t_k}{\alpha'} \right), \\
\label{eq: REP vector}
& \underline{N}_\alpha = \left(N_\alpha(X_1),\ldots,N_\alpha(X_n)\right), \\
& D(\underline{t}\|\underline{N}_\alpha)
=\sum_{k=1}^n t_k \log \left(\frac{t_k}{N_\alpha(X_k)}\right).
\end{align}
In view of \eqref{eq1: BC2} and \eqref{eq2: BC2},
the bound in \eqref{key inequality} holds for every
$\underline{t} \in \Reals_+^n$ such that $\sum_{k=1}^n t_k=1$.
Consequently, the R-EPI in \cite[Theorem~I.1]{BobkovC15} can be tightened
by maximizing the right side of \eqref{key inequality}, leading to
the following optimization problem:
\begin{align} \label{optimize f_0}
\begin{array}{ll}
\text{maximize} &  f_0(\underline{t}) \\
\text{subject to}  & t_k \geq 0 , \quad k \in \{1, \ldots, n\}, \\
                   & \sum_{k=1}^n t_k=1.
\end{array}
\end{align}
Note that the convexity of the function
\begin{align} \label{eq: f}
f(x) = \left(1-\frac{x}{\alpha'}\right) \log \left(1- \frac{x}{\alpha'}\right),
\quad x \in [0, \alpha']
\end{align}
yields that the third term on the right side of \eqref{eq: f0}
is convex in $\underline{t}$. Since the relative entropy
$D(\underline{t}\|\underline{N_\alpha})$ is also convex in
$\underline{t}$, the objective function $f_0$ in \eqref{eq: f0}
is expressed as a difference of two convex functions in $\underline{t}$.
In order to get an analytical closed-form lower bound on the solution
of the optimization problem in \eqref{optimize f_0}, we take the
sub-optimal choice $\underline{t} = \underline{N}_{\alpha}$ (similarly
to the proof \cite[Theorem~I.1]{BobkovC15}) which yields that
$D(\underline{t}\|\underline{N}_\alpha)=0$; however, our proof derives
an improved lower bound on the third term of $f_0(\underline{t})$
which needs to be independent of $\underline{N}_\alpha$. Let
\begin{align} \label{eq: suboptimal t}
\hat{t}_k = N_\alpha(X_k),\quad 1 \leq k \leq n,
\end{align}
then, in view of \eqref{key inequality} and \eqref{eq: suboptimal t},
\begin{align} \label{eq1: sub_optimal REPI}
\log N_\alpha(S_n) & \geq f_0(\underline{\hat{t}})  \\
\label{eq2: sub_optimal REPI}
& = \frac{\log \alpha}{\alpha-1}  +
\alpha' \sum_{k=1}^n\left( 1-\frac{\hat{t}_k}{\alpha'} \right)
\log \left( 1-\frac{\hat{t}_k}{\alpha'} \right).
\end{align}
Due to the convexity of $f$ in \eqref{eq: f}, for all $k \in \{1, \ldots, n\}$,
\begin{align} \label{eq1: convex f}
f(\hat{t}_k) \geq f(x) + f'(x) \, (\hat{t}_k - x).
\end{align}
Choosing $x = \frac1{n}$ in the right side of \eqref{eq1: convex f} yields
\begin{align} \label{eq2: convex f}
\left( 1-\frac{\hat{t}_k}{\alpha'} \right) \log \left( 1-\frac{\hat{t}_k}{\alpha'}
\right) \geq \log \left( 1-\frac1{n\alpha'} \right) + \frac{\log e}{n\alpha'} -
\frac{\hat{t}_k}{\alpha'} \left[ \log e+\log \left( 1-\frac1{n\alpha'} \right)  \right]
\end{align}
and, in view of \eqref{eq: normalized REP} and \eqref{eq: suboptimal t} which
yields $\sum_{k=1}^n \hat{t}_k = 1$, summing over $k \in \{1, \ldots, n\}$ on
both sides of \eqref{eq2: convex f} implies that
\begin{align} \label{LB - third term of BC}
\alpha' \sum_{k=1}^n \left( 1-\frac{\hat{t}_k}{\alpha'} \right)
\log \left( 1-\frac{\hat{t}_k}{\alpha'} \right)
\geq  (n\alpha'-1) \log \left( 1-\frac1{n\alpha'} \right).
\end{align}
Finally, assembling \eqref{eq1: sub_optimal REPI},
\eqref{eq2: sub_optimal REPI} and \eqref{LB - third term of BC} yields
\eqref{eq: equiv. REPI1} with $c_{\alpha}^{(n)}$ in \eqref{eq: c for R-EPI1}
as required.

In the sequel, we no longer assume that condition \eqref{eq: positive N_k}
holds. Define
\begin{align} \label{eq:K_0}
\mathcal{K}_0=\{ k \in \{1, \ldots, n\} \colon N_\alpha(X_k)=0 \},
\end{align}
and note that
\begin{align}
\label{eq:conditinal renyi 1}
h_\alpha(S_n)
&= h_\alpha \left(\sum_{k \notin \mathcal{K}_0}X_k +
\sum_{k \in \mathcal{K}_0}X_k \right) \\[0.2 cm]
\label{eq:conditinal renyi 2}
& \geq h_\alpha \left(\sum_{k \notin \mathcal{K}_0}X_k +
\sum_{k \in \mathcal{K}_0}X_k  \, \Big| \, \{X_k\}_{k \in \mathcal{K}_0}\right) \\[0.2 cm]
\label{eq:conditinal renyi 3}
& = h_\alpha \left(\sum_{k \notin \mathcal{K}_0}X_k\right)
\end{align}
where the conditional R\'{e}nyi entropy is defined according to Arimoto's proposal
in \cite{Arimoto75} (see also \cite[Section~4]{FehrB14}), \eqref{eq:conditinal renyi 2} is due
to the monotonicity property of the conditional R\'{e}nyi entropy (see \cite[Theorem~2]{FehrB14}),
and \eqref{eq:conditinal renyi 3} is due to the independence of $X_1, \ldots, X_n$. Since $N_\alpha(X_k) > 0$
for every $k \notin \mathcal{K}_0$, then from the previous analysis
\begin{align} \label{eq:Th1 K_0^c}
N_\alpha \left(\sum_{k \notin \mathcal{K}_0}X_k\right) \geq c_\alpha^{(l)}
\sum_{k \notin \mathcal{K}_0} N_\alpha(X_k),
\end{align}
where $l=n-|\mathcal{K}_0|$. In view of \eqref{eq: c for R-EPI1},
it can be verified that $c_\alpha^{(n)}$ is monotonically decreasing
in $n$; hence, \eqref{eq:conditinal renyi 3}, \eqref{eq:Th1 K_0^c}
and $c_\alpha^{(l)} \geq c_\alpha^{(n)}$ yield
\begin{align} \label{eq:prrof Th1}
N_\alpha(S_n) \geq c_\alpha^{(n)} \sum_{k=1}^n N_\alpha(X_k).
\end{align}

We now turn to prove Items 1)--4).
\begin{itemize}
\item To prove Item~1), note that \eqref{eq: c for BC}
and \eqref{eq: c for R-EPI1} yield that $c_{\alpha}^{(n)} > c_{\alpha}$ for
all $\alpha > 1$ and $n \in \naturals$.
\item Item~2) holds since from \eqref{eq: c for R-EPI1}
\begin{align} \label{eq: limit of tightened BC const.}
\lim_{n \to \infty} c_\alpha^{(n)} = \tfrac1{e} \, \alpha^{\frac1{1-\alpha}}
\end{align}
where the right side of \eqref{eq: limit of tightened BC const.} coincides
with the constant $c_\alpha$ in \cite[(3)]{BobkovC15} (see \eqref{eq: c for BC}).

\item Item 3) holds since $\alpha \downarrow 1$ yields $\alpha' \to \infty$, which
implies that for every $n \in \naturals$
\begin{align} \label{eq: c at alpha=1}
\lim_{\alpha \downarrow 1} c_\alpha^{(n)} = \lim_{\alpha \downarrow 1} c_{\alpha} = 1.
\end{align}
Hence, by letting $\alpha$ tend to~1, \eqref{eq: REPI1} and
\eqref{eq: c at alpha=1} yield the EPI in \eqref{eq: EPI}.

\item To prove Item~4), note that from \eqref{eq: c for R-EPI1}
\begin{align}
\lim_{\alpha \to \infty} c_\alpha^{(n)} = \left(1 - \frac1{n} \right)^{n-1}
\end{align}
which is monotonically decreasing in $n$ for $n \geq 2$, being equal to
$\tfrac12$ for $n=2$ and $\tfrac1e$ by letting $n$ tend to $\infty$.
Let $X$ be a $d$-dimensional random vector with density $f_X$, and let
\begin{align} \label{eq: MX}
M(X):= \esssup(f_X).
\end{align}
From \eqref{RE at infinity}, \eqref{eq: Renyi entropy power} and
\eqref{eq: MX}, it follows that
\begin{align}
N_{\infty}(X) & := \lim_{\alpha \to \infty} N_{\alpha}(X) \\
\label{eq: N at infinity}
& = M^{-\frac{2}{d}}(X).
\end{align}
By assembling \eqref{eq: REPI1} and \eqref{eq: N at infinity},
it follows that if $X_1, \ldots, X_n$ are independent $d$-dimensional
random vectors with densities then
\begin{align} \label{eq: M-inequality}
M^{-\frac{2}{d}}(S_n) \geq \left(1-\frac{1}{n}\right)^{n-1}
\sum_{k=1}^n M^{-\frac{2}{d}}(X_k).
\end{align}
\end{itemize}
This improves the tightness of the inequality in \cite[Theorem~1]{BobkovC14}
where the coefficient $\left(1-\frac{1}{n}\right)^{n-1}$ on the right side of
\eqref{eq: M-inequality} has been loosened to $\tfrac1e$ (note, however, that
they coincide when $n \to \infty$). For $n=2$, the coefficient $\tfrac12$ on
the right side of \eqref{eq: M-inequality} is tight, and it is achieved when
$X_1$ and $X_2$ are independent random vectors which are uniformly distributed
in the cube $[0,1]^d$ \cite[p.~103]{BobkovC14}.
\end{proof}
\begin{figure}[here!]
\begin{center}
\epsfig{file=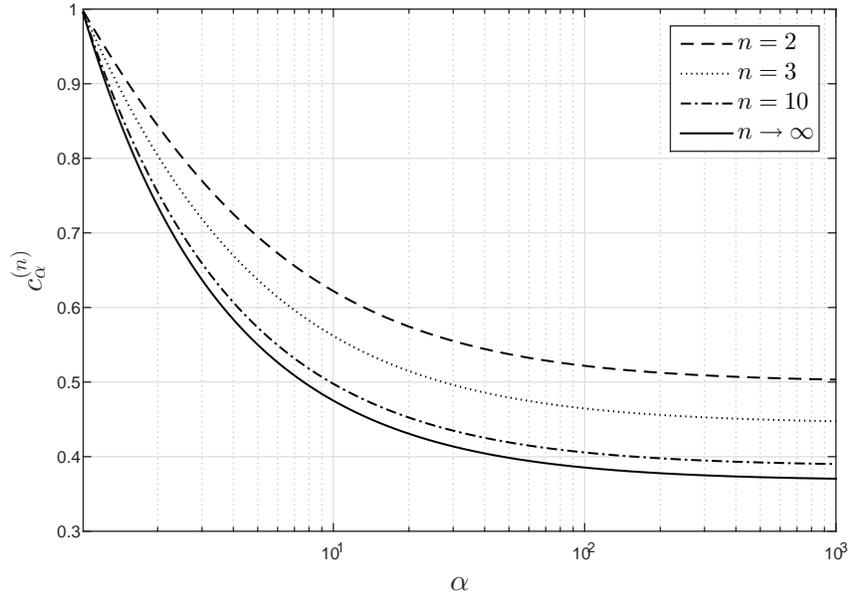,scale=0.65}
\caption{\label{Figure: R-EPI1}
A plot of $c_\alpha^{(n)}$ in \eqref{eq: c for R-EPI1}, as a function
of $\alpha$, for $n=2,3,10$ and $n \to \infty$.}
\end{center}
\end{figure}
Figure~\ref{Figure: R-EPI1} plots $c_\alpha^{(n)}$ as a function
of $\alpha$, for some values of $n$, verifying numerically Items~1)--4)
in Theorem~\ref{theorem: REPI1}.
In \cite[Theorem~I.1]{BobkovC15},
$c_\alpha^{(n)}$ is independent of $n$, and it is equal to $c_\alpha$ in
\eqref{Intro: R-EPI} which is the limit of $c_\alpha^{(n)}$
in \eqref{eq: c for R-EPI1} by letting $n \to \infty$ (the solid curve in
Figure~\ref{Figure: R-EPI1}).

\begin{remark}
For independent random variables $\{X_k\}_{k=1}^n$ with densities on $\Reals$,
the result in \eqref{eq: M-inequality} with $d=1$ can be strengthened to (see
\cite[p.~105]{BobkovC14} and \cite{Rogozin})
\begin{align}  \label{eq2: M-inequality}
\frac1{M^2(S_n)} \geq \tfrac12 \sum_{k=1}^n \frac1{M^2(X_k)}
\end{align}
where $S_n := \sum_{k=1}^n X_k$. Note that \eqref{eq: M-inequality}
and \eqref{eq2: M-inequality} coincide if $n=2$ and $d=1$.
\end{remark}

\begin{example}
Let $X$ and $Y$ be $d$-dimensional random vectors with densities $f_X$ and $f_Y$,
respectively, and assume that the entries of $X$ are i.i.d. as well as those of
$Y$. Let $X_1$, $X_2$, $Y_1$, $Y_2$ be independent $d$-dimensional random vectors
where $X_1, X_2$ are independent copies of $X$, and $Y_1, Y_2$ are independent
copies of $Y$. Assume that
\begin{align}
\label{eq: Pr X1=X2 Y1=Y2}
\begin{split}
& \Pr[X_{1,k} = X_{2,k}] = \alpha, \\
& \Pr[Y_{1,k} = Y_{2,k}] = \beta
\end{split}
\end{align}
for all $k \in \{1, \ldots, n\}$. We wish to obtain an upper bound on the
probability that $X_1+Y_1$ and $X_2+Y_2$ are equal. From \eqref{eq: Renyi entropy},
\eqref{eq: Renyi entropy power} (with $\alpha=2$), and \eqref{eq: Pr X1=X2 Y1=Y2}
\begin{align}
& N_2(X) = \exp\left(\tfrac2d \, h_2(X)\right)  \\
& \hspace*{1.1cm} = \left( \int_{\Reals^d} f_X^2(\underline{x}) \, \mathrm{d}\underline{x} \right)^{-\tfrac2d} \\
& \hspace*{1.1cm} = \probability^{-\frac2d}[X_1 = X_2] \\
& \hspace*{1.1cm} = \prod_{k=1}^d \probability^{-\frac2d}[X_{1,k} = X_{2,k}] \\
& \hspace*{1.1cm} = \alpha^{-2}, \label{N2 X} \\[0.1cm]
& N_2(Y) = \beta^{-2}, \label{N2 Y} \\
\label{N2 X+Y}
& N_2(X+Y) = \probability^{-\frac2d}[X_1+Y_1 = X_2+Y_2].
\end{align}
Assembling \eqref{eq: REPI1} with $n = \alpha = 2$, \eqref{N2 X}, \eqref{N2 Y} and \eqref{N2 X+Y} yield
\begin{align} \label{eq: Pr X+Y}
\probability(X_1+Y_1 = X_2+Y_2) \leq
\left( \tfrac{27}{32} \, \bigl(\alpha^{-2} + \beta^{-2}\bigr) \right)^{-\frac{d}{2}}.
\end{align}
The factor $\tfrac{27}{32}$ on the base of the exponent on the right side of \eqref{eq: Pr X+Y},
instead of the looser factor $c_2 = \tfrac2e$ which follows from \eqref{eq: c for BC} with $\alpha=2$
(see \cite[Theorem~I.1]{BobkovC15}), improves the exponential decay rate of the upper bound in
\eqref{eq: Pr X+Y} as a function of the dimension $d$. The optimal bound
has to be with a coefficient of $\bigl(\alpha^{-2} + \beta^{-2}\bigr)$ on the base of the exponent
in the right side of \eqref{eq: Pr X+Y} which is less than or equal to~1; this can be verified since
if $X$ and $Y$ are independent Gaussian random variables, then
\begin{align} \label{eq:N2 X+Y Gaussians}
N_2(X+Y)=N_2(X)+N_2(Y),
\end{align}
so,
\begin{align} \label{eq:Pr X+Y Gaussians}
\probability(X_1+Y_1 = X_2+Y_2)= \bigl( \alpha^{-2} + \beta^{-2} \bigr) ^{-\tfrac{d}{2}}.
\end{align}
This provides a reference for comparing the exponential decay which is implied by
$c_2$ in \eqref{eq: c for BC}, $c_2^{(2)}$ in \eqref{eq: REPI1}, and the case
where $X$ and $Y$ are independent Gaussian random variables:
\begin{align} \label{eq:comparing}
\frac{2}{e} < \frac{27}{32} < 1.
\end{align}
\end{example}

\section{A Further tightening of the R-EPI}
\label{section: further tightening of R-EPI}
\subsection{A Tightened R-EPI for $n \geq 2$}
\label{subsection: REPI2 - general n}
In the following, we wish to tighten the R-EPI in Theorem~\ref{theorem: REPI1}
for an arbitrary $n \geq 2$. It is first demonstrated that a reduction of the
optimization problem in \eqref{optimize f_0} to $n-1$ variables (recall that
$\sum_{k=1}^n t_k = 1$) leads to a convex optimization problem. This convexity
result is established by a non-trivial use of Fact~\ref{fact: Bunch} in
Section~\ref{subsection:Pre - Rank-One Modification} (see \cite{Bunch}),
and it is also shown that the reduction of the optimization problem in \eqref{optimize f_0}
from $n$ to $n-1$ variables is essential for its convexity. Consequently, the
convex optimization problem is handled by solving the corresponding Karush-Kuhn-Tucker
(KKT) equations. If $n=2$, their solution leads to a closed-form expression which yields the
R-EPI in Corollary~\ref{proposition: improved bound n=2}. For $n>2$, no solution is provided
in closed form; nevertheless, an efficient algorithm is introduced for solving the KKT
equations for an arbitrary $n>2$, and the improvement in the tightness of the new R-EPI is
exemplified numerically in comparison to the bounds in \cite{BercherVignat}, \cite{BobkovC15}
and Theorem~\ref{theorem: REPI1}.

\subsubsection{The optimization problem in \eqref{optimize f_0}}
\label{subsection: optimization problem with f_0}
In view of \eqref{eq: f0}--\eqref{optimize f_0}, the maximization problem in
\eqref{optimize f_0} can be expressed in the form
\begin{align} \label{opt. prob. with n variables}
\begin{array}{ll}
\text{maximize} & f_0(\underline{t}) = \sum_{k=1}^n g(t_k) + \sum_{k=1}^n t_k \log N_k
+ \frac{\log \alpha}{\alpha-1} \\[0.1 cm]
\text{subject to}  & \underline{t} \in \mathcal{P}^n
\end{array}
\end{align}
where
\begin{align} \label{eq: g}
& g(x)=(\alpha'-x)\log \left( 1-\frac{x}{\alpha'} \right) - x \log x, \quad x \in [0,1]\\
\label{eq: N_alpha entries}
& N_k = N_{\alpha}(X_k), \quad k \in \{1, \ldots, n\}
\end{align}
(for simplicity of notation, the dependence of $g$ and $N_k$ in $\alpha$ has been suppressed in
\eqref{opt. prob. with n variables}), and $\mathcal{P}^n$ is the probability simplex
\begin{align} \label{eq: n-dim simplex}
\mathcal{P}^n = \left\{\underline{t} \in \Reals^n \colon t_k \geq 0,\, \sum_{k=1}^n t_k = 1\right\}.
\end{align}
The term $ \sum_{k=1}^n t_k \log N_k$ on the right side of \eqref{opt. prob. with n variables}
is linear in $\underline{t}$, thus the concavity of $f_0$ in $\underline{t}$ is only affected by the
term $\sum_{k=1}^n g(t_k)$. Since $g''(x) = \frac{2x-\alpha'}{x(\alpha'-x)}$ where $x \in [0,1]$,
if $\alpha' \geq 2$, then $g$ is concave on the interval $[0,1]$. If $\alpha' \in (1,2)$ (i.e.,
if $\alpha \in (2, \infty)$) then $g$ is not concave on the interval $[0,1]$; it is only concave on
$[0, \tfrac{\alpha'}{2}]$, and it is convex on $[\tfrac{\alpha'}{2}, 1]$. Hence, as a maximization
problem over the variables $t_1, \ldots, t_n$, the objective function $f_0$ in
\eqref{opt. prob. with n variables} is not concave if $\alpha > 2$.

\subsubsection{A reduction of the optimization problem in \eqref{optimize f_0} to $n-1$ variables}
In view of \eqref{eq: n-dim simplex}, the substitution
\begin{align}
\label{eq: tn}
t_n = 1 - \sum_{k=1}^{n-1} t_k
\end{align}
transforms the maximization problem in \eqref{opt. prob. with n variables} to the following
equivalent problem:
\begin{align} \label{opt. prob. with n-1 variables}
\begin{array}{ll}
\text{maximize} & f(t_1,\ldots,t_{n-1})\\[0.1cm]
\text{subject to}  & \underline{t} \in \mathcal{D}^{n-1}
\end{array}
\end{align}
where
\begin{align} \label{eq: f function}
f(t_1,\ldots,t_{n-1}) = f_0 \left(t_1, \ldots, t_{n-1}, 1-\sum_{k=1}^{n-1} t_k \right)
\end{align}
and $\mathcal{D}^{n-1}$ is the polyhedron
\begin{align} \label{eq: polyhedron}
\mathcal{D}^{n-1} = \left\{(t_1, \ldots,t_{n-1} ) \colon t_k \geq 0,
\; \sum_{k=1}^{n-1}t_k \leq 1 \right\}.
\end{align}

\subsubsection{Proving the convexity of the optimization problem in
\eqref{opt. prob. with n-1 variables}}
We wish to show that the objective function $f$ of the optimization problem
in \eqref{opt. prob. with n-1 variables} is concave, i.e., it is required to assert
that all the eigenvalues of the Hessian matrix $\nabla^2 f$ are non-positive.

Eqs. \eqref{opt. prob. with n variables} and \eqref{eq: f function} yield
\begin{align} \label{eq: rewriting f}
\begin{split}
& f(t_1,\ldots,t_{n-1}) \\
& = \sum_{k=1}^{n-1} g(t_k) + g\left(1-\sum_{k=1}^{n-1} t_k\right) \\
& \hspace*{0.3cm} + \sum_{k=1}^{n-1} t_k \log N_k +
\left(1-\sum_{k=1}^{n-1} t_k\right) \log N_n
+ \frac{\log \alpha}{\alpha-1}.
\end{split}
\end{align}
Let
\begin{align} \label{eq: q function}
q(x) = g''(x) = \frac{2x-\alpha'}{x(\alpha'-x)}, \quad x \in [0,1]
\end{align}
then, in view of \eqref{eq: rewriting f} and \eqref{eq: q function},
for all $(t_1, \ldots, t_{n-1}) \in \mathcal{D}^{n-1}$
\begin{align} \label{eq1: Hessian of f}
\nabla^2 f(t_1,\ldots,t_{n-1}) &=
\begin{pmatrix}
  q(t_1) & 0 & \cdots & 0 \\
  0 & q(t_2) & \cdots & 0 \\
  \vdots  & \vdots  & \ddots & \vdots  \\
  0 & 0 & \cdots & q(t_{n-1})
\end{pmatrix} + q\left(1-\sum_{k=1}^{n-1}t_k\right)
\begin{pmatrix}
  1 & 1 & \cdots & 1 \\
  1 & 1 & \cdots & 1 \\
  \vdots  & \vdots  & \ddots & \vdots  \\
  1 & 1 & \cdots & 1
\end{pmatrix} \nonumber \\
&= D + \rho \, \underline{1} \, \underline{1}^T
\end{align}
where
\begin{align} \label{eq: D and rho}
\begin{split}
&D =  \text{diag}(q(t_1), \ldots, q(t_{n-1})), \\
&\rho = q\left(1-\sum_{k=1}^{n-1}t_k\right).
\end{split}
\end{align}
Recall that if $\alpha' \in [2,\infty)$ then $f_0(t_1,\ldots,t_n)$ is
concave in $\mathcal{P}^n$, hence, so is $f(t_1,\ldots,t_{n-1})$ in
$D^{n-1}$. We therefore need only to focus on the case where $\alpha' \in (1,2)$
(i.e., $\alpha \in (2, \infty)$).

\begin{proposition} \label{proposition: concavity of f}
For every $\alpha' \in (1,2)$, the function $f \colon \mathcal{D}^{n-1} \to \Reals$
in \eqref{eq: rewriting f} is concave.
\end{proposition}
\begin{proof}
See Appendix~\ref{appendix: concavity of f}.
\end{proof}

\subsubsection{Solution of the convex optimization problem in
\eqref{opt. prob. with n-1 variables}}
\label{subsection: Solving the convex optimization problem}
In the following, we solve the convex optimization problem in \eqref{opt. prob. with n-1 variables}
via the Lagrange duality and KKT conditions (see, e.g., \cite[Chapter~5]{BoydV_book}).
Since the problem is invariant to permutations of the entries of $X = (X_1, \ldots, X_n)$, it can be
assumed without any loss of generality that the last term of the vector $\underline{N}_{\alpha}$ in
\eqref{eq: REP vector} is maximal, i.e.,
\begin{align} \label{eq: last entry of N is maximal}
N_\alpha(X_k) \leq N_\alpha(X_n), \quad k \in \{1, \ldots, n-1\}.
\end{align}
Moreover, it is assumed that
\begin{align} \label{positive}
N_\alpha(X_n)>0.
\end{align}
The possibility that $N_\alpha(X_n)=0$ leads to a trivial bound since from
\eqref{eq: last entry of N is maximal}, it follows that $N_\alpha(X_k)=0$
for every $k \in \{1,\ldots,n\}$; this makes the right side of \eqref{Intro: R-EPI}
be equal to zero, while its left side is always non-negative. Let
\begin{align}
\label{eq: c_k}
& c_k = \frac{N_\alpha(X_k)}{N_\alpha(X_n)},
\quad \quad k \in \{1, \ldots, n-1\}.
\end{align}
From \eqref{eq: last entry of N is maximal}--\eqref{eq: c_k}, the sequence
$\{c_k\}_{k=1}^{n-1}$ satisfies
\begin{align} \label{eq: sequence c is increasing and bounded}
0 \leq c_k \leq 1, \quad k \in \{1, \ldots, n-1\}.
\end{align}
Let $t_n$ be defined as in \eqref{eq: tn}. Appendix~\ref{appendix: Lagrangian}
provides the technical details which are related to the solution of the convex
optimization problem in \eqref{opt. prob. with n-1 variables}
via the Lagrange duality and KKT conditions (note that strong duality holds here).
The resulting simplified set of constraints which follow from the KKT conditions (see
Appendix~\ref{appendix: Lagrangian}) is given by
\begin{align}
\label{eq1s: KKT}
& t_k (\alpha'-t_k)= c_k t_n (\alpha'-t_n), \quad k \in \{1, \ldots, n-1\} \\
\label{eq2s: KKT}
& \sum_{k=1}^n t_k = 1 \\
\label{eq3s: KKT}
& t_k  \geq 0, \quad k \in \{1, \ldots, n\}
\end{align}
with the variables  $\underline{t}$ in
\eqref{eq1s: KKT}--\eqref{eq3s: KKT}.

Note that if $N_{\alpha}(X_k)$ is independent of $k$ then,
from \eqref{eq: c_k}, $c_k = 1$ for all $k \in \{1, \ldots, n-1\}$.
Hence, from \eqref{eq1s: KKT} and \eqref{eq2s: KKT}, it follows that
$t_1 = \ldots = t_n = \frac1n$ (note that the other
possibility where $t_k = \alpha' - t_n$ for some $k \in \{1, \ldots, n-1\}$
contradicts \eqref{eq2s: KKT} and \eqref{eq3s: KKT} since in this case
$\sum_{j=1}^n t_j \geq t_k + t_n = \alpha' > 1$). This implies that the
selection of the $t_k$'s in the proof of Theorem~\ref{theorem: REPI1}
is optimal when all the entries of the vector
$\underline{N}_{\alpha}$ are equal; therefore, the R-EPI considered
here improves the bound in Theorem~\ref{theorem: REPI1}
only when $N_{\alpha}(X_k)$ depends on the index $k$.

In the general case, \eqref{eq1s: KKT} yields a quadratic
equation for $t_k$ whose solutions are given by
\begin{align} \label{eq: possible t_k's}
t_k = \tfrac12 \, \left(\alpha' \pm \sqrt{\alpha'^{\, 2}
- 4 c_k t_n (\alpha'-t_n)} \right)
\end{align}
with $\alpha'=\frac{\alpha}{\alpha-1}$.
The possibility of the positive sign in the right side of
\eqref{eq: possible t_k's} is rejected since in this case
$t_n + t_k \geq \alpha' > 1$, which violates
\eqref{eq2s: KKT}. Hence, from \eqref{eq: possible t_k's},
for all $k \in \{1, \ldots, n-1\}$
\begin{align} \label{eq: the t_k solution}
t_k = \psi_{k,\alpha}(t_n)
\end{align}
where we define
\begin{align} \label{eq: psi_k}
&\psi_{k,\alpha}(x) = \tfrac12 \, \left(\alpha' - \sqrt{\alpha'^{\, 2}
- 4 c_k \, x (\alpha'-x)}\right),
\quad x \in [0,1].
\end{align}
In view of \eqref{eq2s: KKT} and \eqref{eq: the t_k solution}, one
first calculates $t_n \in [0,1]$ by numerically solving the equation
\begin{align} \label{eq for tn}
t_n + \sum_{k=1}^{n-1} \psi_{k,\alpha}(t_n) = 1.
\end{align}
The existence and uniqueness of a solution of \eqref{eq for tn}
is proved in Appendix~\ref{appendix: Exist and Uniq}. Once we compute
$t_n$, all $t_k$'s for $k \in \{1, \ldots, n-1\}$ are
computed from \eqref{eq: the t_k solution}.
Finally, the substitution of $t_1, \ldots, t_n$ in
the right side of \eqref{key inequality} enables to calculate the improved
R-EPI in \eqref{key inequality}, i.e.,
\begin{align} \label{eq: tightest REPI}
N_\alpha \left( \sum_{k=1}^n X_k \right)
\geq \exp\bigl(f_0(t_1, \ldots, t_n)\bigr) \, \sum_{k=1}^n N_\alpha(X_k)
\end{align}
with $f_0$ in \eqref{eq: f0}.

Note that due to the optimal selection of the vector
$\underline{t} = (t_1, \ldots, t_n)$ in \eqref{eq: tightest REPI},
the R-EPI in this section provides an improvement
over the R-EPI in Theorem~\ref{theorem: REPI1} whenever
$N_{\alpha}(X_k)$ is not fixed as a function of the index $k$.
This leads to the following result:

\begin{theorem}
\label{theorem: tightest REPI}
Let $X_1, \ldots, X_n$ be independent random vectors with probability densities
defined on $\Reals^d$, let $N_{\alpha}(X_1), \ldots, N_{\alpha}(X_n)$ be their
respective R\'{e}nyi entropy powers of order $\alpha>1$, and let
$\alpha' = \frac{\alpha}{\alpha-1}$. Let the indices of $X_1, \ldots, X_n$
be set such that $N_{\alpha}(X_n)$ is maximal, and let
\begin{enumerate}
\item $\{c_k\}_{k=1}^{n-1}$ be the sequence defined in \eqref{eq: c_k}; \\[-0.4cm]
\item $t_n \in [0,1]$ be the unique solution of \eqref{eq for tn}; \\[-0.4cm]
\item $\{t_k\}_{k=1}^{n-1}$ be given in \eqref{eq: the t_k solution} and \eqref{eq: psi_k}.
\end{enumerate}
Then, the R-EPI in \eqref{eq: tightest REPI} holds with $f_0$ in \eqref{eq: f0}, and it
satisfies the following properties:
\begin{enumerate}
\item It improves the R-EPI in Theorem~\ref{theorem: REPI1} unless $N_{\alpha}(X_k)$ is
independent of $k$ (consequently, it also improves the R-EPI in \cite[Theorem~1]{BobkovC15});
if $N_{\alpha}(X_k)$ is independent of $k$, then the two R-EPIs in Theorem~\ref{theorem: REPI1}
and \eqref{eq: tightest REPI} coincide.
\item It improves the Bercher-Vignat (BV) bound in \cite{BercherVignat} which states that
\begin{align} \label{BV bound - general n}
N_{\alpha}\left( \sum_{k=1}^n X_k \right) \geq
\max \bigl\{ N_{\alpha}(X_1), \ldots, N_{\alpha}(X_n) \bigr\}
\end{align}
and the bounds in \eqref{eq: tightest REPI} and \eqref{BV bound - general n} asymptotically
coincide as $\alpha \to \infty$ if and only if
\begin{align} \label{Tightest REPI meets BV}
\sum_{k=1}^{n-1}N_{\infty}(X_k) \leq N_{\infty}(X_n)
\end{align}
where $N_{\infty}(X)$ is defined in \eqref{eq: N at infinity}.
\item For $n=2$, it is expressed in a closed form (see Corollary~\ref{proposition: improved bound n=2}).
\item It coincides with the EPI and the two R-EPIs in \cite[Theorem~1]{BobkovC15}
and Theorem~\ref{theorem: REPI1} as $\alpha \downarrow 1$.
\end{enumerate}
\end{theorem}

\begin{proof}
The proof of the R-EPI in \eqref{eq: tightest REPI} is provided earlier in this section
with some additional details in Appendices~\ref{appendix: concavity of f}--\ref{appendix: Tightest REPI meets BV}.
In view of the this analysis:
\begin{itemize}
\item
Item~1) holds since the proof of the R-EPI in Theorem~\ref{theorem: REPI1} relies in general
on a sub-optimal choice of the vector $\underline{t}$ in \eqref{eq: suboptimal t}, whereas
it is set to be optimal in the proof of Theorem~\ref{theorem: tightest REPI} in
\eqref{eq: the t_k solution}--\eqref{eq for tn}.
Suppose, however, that $N_{\alpha}(X_k)$ is independent of the index $k$; in the latter case,
the selection of the vector $\underline{t}$ in the proof of Theorem~\ref{theorem: REPI1} (see
\eqref{eq: suboptimal t}) reduces to $\underline{t} = \bigl( \tfrac1n, \ldots, \tfrac1n \bigr)$,
which turns to be optimal in the sense of achieving the maximum of the objective function in
\eqref{eq: rewriting f}.
\item Item~2) holds since the selection of $\underline{t}$ in the right side of \eqref{key inequality}
with $t_k=1$ and $t_i = 0$ for all $i \neq k$ yields
\begin{align}  \label{temp}
N_{\alpha}\left( \sum_{k=1}^n X_k \right) \geq N_{\alpha}(X_k)
\end{align}
which then leads to \eqref{BV bound - general n} by a maximization of the right side of
\eqref{temp} over $k \in \{1, \ldots, n\}$. Appendix~\ref{appendix: Tightest REPI meets BV}
proves that the bounds in \eqref{eq: tightest REPI} and \eqref{BV bound - general n}
asymptotically coincide as $\alpha \to \infty$ if and only if the condition in
\eqref{Tightest REPI meets BV} holds.
\item Item~3) is proved in Section~\ref{subsection: The Two Element Case}.
\item Item~4) holds since the R-EPI obtained in Theorem~\ref{theorem: tightest REPI} is at least
as tight as the BC bound in \cite[Theorem~1]{BobkovC15}; the latter coincides with the EPI
as we let $\alpha$ tend to 1 (recall that, from \eqref{eq: c for BC}, $\lim_{\alpha \downarrow 1} c_{\alpha} = 1$)
which is known to be tight for Gaussian random vectors with proportional covariances.
\end{itemize}
\end{proof}

\begin{remark} \label{Remark:tightness Th2}
The R-EPI in Theorem~\ref{theorem: tightest REPI} provides the tightest R-EPI known to date
for $\alpha \in (1, \infty)$. Nevertheless, it is still not tight for $\alpha \in (1, \infty)$
since at least one of the inequalities involved in the derivation of \eqref{eq: R-EPI BC15}
(see Appendix~\ref{appendix: BC15}) is loose. These include the sharpened Young's inequality
in \eqref{eq: Young's ineq. 3}, and \eqref{inequality:Pre -  Holder for densities}.
The former inequality holds with equality only for Gaussians, whereas the latter inequality
holds with equality only for a uniformly distributed random variable (note that in the latter
case, the R\'enyi entropy is independent of its order). For $\alpha=\infty$ and $n=2$, the
sharpened Young's inequality \eqref{eq: Young's ineq. 1} reduces to
\begin{align} \label{eq:Young alpha=inf,n=2}
\|f \ast g\|_\infty \leq \|f\|_p \, \|g\|_{p'}
\end{align}
where $p>1$ and $p'=\frac{p}{p-1}$. Equality holds in \eqref{eq:Young alpha=inf,n=2} if $f$ and $g$
are scaled versions of a uniform distribution on the same convex set, which is also the same condition
for tightness of \eqref{inequality:Pre -  Holder for densities}; this is consistent with our conclusion
that the R-EPIs in Theorems~\ref{theorem: REPI1} and~\ref{theorem: tightest REPI} are, however, asymptotically
tight for $n=2$ by letting $\alpha \to \infty$.
\end{remark}

\vspace*{0.1cm}
\begin{figure}[here!]
\begin{center}
\epsfig{file=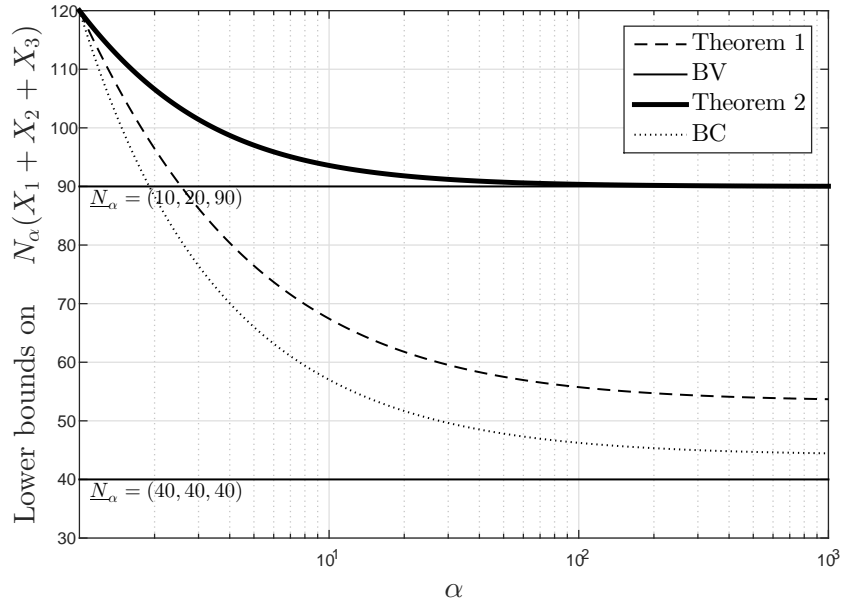,scale=0.65}
\caption{ \label{Figure: n=3}
A comparison of the R\'{e}nyi entropy power inequalities for $n=3$ independent random vectors
according to \cite{BobkovC15} (BC), \cite{BercherVignat} (BV),
Theorem~\ref{theorem: REPI1} and the tightest bound in Theorem~\ref{theorem: tightest REPI}.
The bounds refer to the two cases where $(N_\alpha(X_1), N_\alpha(X_2), N_\alpha(X_3))=(40,40,40)$
or $(10,20,90)$ (in both cases, the sum of the entries is 120; in the former case, the condition
in \eqref{Tightest REPI meets BV} does not hold, while in the latter it does).}
\end{center}
\end{figure}

Figure~\ref{Figure: n=3} compares the two R-EPIs in Theorems~\ref{theorem: REPI1}
and~\ref{theorem: tightest REPI} with those in \cite{BobkovC15}
(see \eqref{eq: c for BC}) and \cite{BercherVignat} (see \eqref{BV bound - general n})
for $n=3$ independent random vectors; the abbreviations 'BC' and 'BV' stand, respectively,
for the latter two bounds. Recall that the four bounds are independent of the dimension $d$
of the random vectors, and they are plotted in Figure~\ref{Figure: n=3} for symmetric
and asymmetric cases where $(N_\alpha(X_1), N_\alpha(X_2), N_\alpha(X_3)) = (40, 40, 40)$
and $(10, 20, 90)$, respectively (note that in both cases, the sum of the entries is equal to 120).
In the former case, for every $\alpha > 1$, Theorem~\ref{theorem: tightest REPI} provides a
lower bound on $N_{\alpha}(X_1 + X_2 + X_3)$ which is tighter than those in
\cite{BercherVignat} and \cite{BobkovC15}; furthermore, in this special case where $N_{\alpha}(X_k)$
is independent of the index $k$, the bounds in Theorems~\ref{theorem: REPI1} and~\ref{theorem: tightest REPI}
coincide. In the asymmetric case, however, where $(N_\alpha(X_1), N_\alpha(X_2), N_\alpha(X_3)) =
(10, 20, 90)$, the bound in Theorem~\ref{theorem: tightest REPI} suggests a significant improvement
over the bound in Theorem~\ref{theorem: REPI1} due to the sub-optimality of the choice of the vector
$\underline{t}$ in the proof of Theorem~\ref{theorem: REPI1} in comparison to its optimal choice
in Theorem~\ref{theorem: tightest REPI}. As it is shown in Figure~\ref{Figure: n=3} and supported
by Item~2) of Theorem~\ref{theorem: tightest REPI}, the bound in this theorem asymptotically
coincides with the BV bound (by letting $\alpha \to \infty$) in the considered asymmetric case;
however, for every $\alpha \in (1, \infty)$, the bound in Theorem~\ref{theorem: tightest REPI}
is advantageous over the BV bound. It is also shown in Figure~\ref{Figure: n=3} that in this asymmetric case,
the BV bound is advantageous over our bound in Theorem~\ref{theorem: REPI1} for sufficiently large $\alpha$;
this observation emphasizes the significance of the optimization of the vector $\underline{t}$ in the proof
of Theorem~\ref{theorem: tightest REPI}, yielding the tightest R-EPI known to date for $\alpha > 1$.
Finally, as it is shown in Figure~\ref{Figure: n=3}, the R-EPIs of Theorems~\ref{theorem: REPI1}
and~\ref{theorem: tightest REPI}, as well as \cite[Theorem~1]{BobkovC15}, coincide with the EPI as
we let $\alpha$ tend to~1 (from above).

\subsection{A Tightened R-EPI for $n=2$}
\label{subsection: The Two Element Case}
We derive in the following a closed-form expression of the R-EPI in Theorem~\ref{theorem: tightest REPI}
for $n=2$ independent random vectors.
In the sequel, we make use of the binary relative entropy function which is defined to be the
continuous extension to $[0,1]^2$ of
\begin{align} \label{eq: binary RE}
d(x\|y) = x \log \left(\frac{x}{y}\right) + (1-x) \log \left(\frac{1-x}{1-y}\right).
\end{align}

\begin{corollary}
\label{proposition: improved bound n=2}
Let $X_1$ and $X_2$ be independent random vectors with densities defined on $\Reals^d$,
let $N_{\alpha}(X_1)$, $N_{\alpha}(X_2)$ be their R\'{e}nyi entropy powers of order
$\alpha>1$, and assume without any loss of generality that $N_{\alpha}(X_1) \leq N_{\alpha}(X_2)$.
Let
\begin{align}
& \alpha' = \frac{\alpha}{\alpha-1}, \\[0.1cm]
\label{eq: beta for n=2}
& \beta_{\alpha} = \frac{N_{\alpha}(X_1)}{N_{\alpha}(X_2)}, \\[0.1cm]
& \label{eq: t* for n=2}
t_{\alpha} = \left\{
\begin{array}{cl}
\frac{\alpha' (\beta_\alpha+1) - 2\beta_\alpha - \sqrt{(\alpha' \,  (\beta_\alpha+1))^2
- 8\alpha' \beta_\alpha + 4\beta_\alpha} }{2(1-\beta_\alpha)}
&\; \mbox{if } \beta_\alpha < 1,\\[0.2cm]
\tfrac1{2} &\; \mbox{if } \beta_\alpha = 1.
\end{array}
\right.
\end{align}
Then, the following R-EPI holds:
\begin{align} \label{eq: improved R-EPI n=2}
N_{\alpha}(X_1 + X_2) \geq c_\alpha \; \bigl( N_{\alpha}(X_1) + N_{\alpha}(X_2) \bigr)
\end{align}
with
\begin{align}
\label{eq: c for n=2}
c_\alpha &= \alpha^{\frac1{\alpha-1}} \,
\exp\left(-d\Bigl(t_{\alpha} \, \bigl\| \, \frac{\beta_{\alpha}}{\beta_{\alpha}+1}
\Bigr)\right) \, \left(1-\frac{t_{\alpha}}{\alpha'}\right)^{\alpha'-t_{\alpha}}
\left(1-\frac{1-t_{\alpha}}{\alpha'}\right)^{\alpha'-1+t_{\alpha}}.
\end{align}
The R-EPI in \eqref{eq: improved R-EPI n=2} satisfies Items~1)--4) of
Theorem~\ref{theorem: tightest REPI}; specifically, by letting $\alpha \to \infty$,
the lower bound on $N_{\alpha}(X_1+X_2)$ tends to $N_{\infty}(X_2)$, which
asymptotically coincides with the BV bound in \cite{BercherVignat}.
\end{corollary}

\begin{proof}
Due the constraints in \eqref{opt. prob. with n variables}, the vector $\underline{t}$
can be parameterized in the form $\underline{t} = (t, 1-t)$ for $t \in [0,1]$; due to
the normalization of the vector $\underline{N}_{\alpha} = (N_\alpha(X_1), N_\alpha(X_2))$
in \eqref{eq: normalized REP}, then
\begin{align}
& \underline{N}_\alpha = \left( \tfrac{\beta_\alpha}{1+\beta_\alpha}, \tfrac1{1+\beta_\alpha} \right)
\end{align}
and, by \eqref{eq: f0}, the maximization in \eqref{opt. prob. with n variables} is transformed to
\begin{align} \label{eq: max f_0 for n=2}
\begin{split}
\underset{t \in [0,1]}{\text{maximize}} & \left\{\frac{\log \alpha}{\alpha-1}
- t \log \bigl( (1+\beta_\alpha) t \bigr) - (1-t) \log \left( \frac{(1+\beta_\alpha)
(1-t)}{\beta_\alpha} \right) \right. \\[0.1cm]
& \left. + \alpha' \left[ \left(1-\frac{t}{\alpha'}\right)
\log\left(1-\frac{t}{\alpha'}\right)
+ \left(1-\frac{1-t}{\alpha'}\right) \log\left(1-\frac{1-t}{\alpha'}\right)
\right] \right\}.
\end{split}
\end{align}
It can be verified that the objective function in \eqref{eq: max f_0 for n=2}
is concave on $[0,1]$, it has a right derivative at $t=0$ which is equal to $+\infty$,
and a left derivative at $t=1$ which is equal to $-\infty$. This implies that the
maximization of the objective function over $[0,1]$ is attained at an interior point
of this interval. The optimized value of $t$ is obtained by setting the derivative
of this objective function to zero, leading to the equation
\begin{align} \label{eq:f'=0}
\log \left( \tfrac{(1-t)\beta_\alpha}{t} \right)-\log \left( \tfrac{\alpha'-t}{\alpha'-1+t} \right)=0.
\end{align}
Eq.~\eqref{eq:f'=0} can be expressed as a quadratic equation whose solution is given in \eqref{eq: t* for n=2}.
Substituting the optimized value $t=t_\alpha$ in \eqref{eq: t* for n=2} into the objective
function on the right side of \eqref{eq: max f_0 for n=2} leads to the closed-form
solution of the optimization problem in \eqref{opt. prob. with n variables} for $n=2$. Hence, under
the assumption in \eqref{eq: normalized REP} where $N_{\alpha}(X_1) + N_{\alpha}(X_2)=1$,
straightforward algebra yields that
\begin{align}
\label{eq: REPI for n=2; normalized case}
N_{\alpha}(X_1 + X_2) \geq c_\alpha
\end{align}
where $c_\alpha$ is given in \eqref{eq: c for n=2}; the relaxation of
this assumption requires the multiplication of the right side of
\eqref{eq: REPI for n=2; normalized case} by $N_{\alpha}(X_1) + N_{\alpha}(X_2)$
(due to the homogeneity of the R\'{e}nyi entropy power, see \eqref{eq: REP}).
Note that, for $n=2$, the condition in \eqref{Tightest REPI meets BV} becomes vacuous (since,
by assumption, $N_{\infty}(X_1) \leq N_{\infty}(X_2)$) which implies that
the bound in \eqref{eq: improved R-EPI n=2} asymptotically coincides with
the BV bound when $\alpha \to \infty$.
\end{proof}

\section{Example: The R\'{e}nyi Entropy Difference Between Data and its Filtering}
\label{section: example}
Let $\{\underline{X}(n)\}$ be i.i.d. $d$-dimensional random vectors (the
entries of the vector $\underline{X}(n)$ need not be independent), with arbitrary
densities on $\Reals^d$. Let
\begin{align}
\underline{Y}(n) = \sum_{k=0}^{L-1} {\bf{H}}_k \, \underline{X}(n-k)
\end{align}
be the filtered data at the output of a finite impulse response (FIR) filter where
${\bf{H}}_0, \ldots, {\bf{H}}_{L-1}$ are fixed non-singular $d \times d$ matrices.

In the following, the tightness of several R-EPIs is exemplified by obtaining universal
lower bounds on the difference
$h_{\alpha}\bigl(\underline{Y}(n)\bigr) - h_{\alpha}\bigl(\underline{X}(n)\bigr)$,
being also compared with the actual value of this difference when the i.i.d. inputs
are $d$-dimensional Gaussian random vectors with i.i.d. entries.

For $k \in \{0, \ldots, L-1\}$ and every $n$, we have
\begin{align} \label{100}
h_{\alpha}\bigl({\bf{H}}_k \, \underline{X}(n-k)\bigr)
= h_{\alpha}\bigl(\underline{X}(n)\bigr) + \log \bigl| \det({\bf{H}}_k) \bigr|
\end{align}
and
\begin{align}
& N_{\alpha}\bigl({\bf{H}}_k \, \underline{X}(n-k)\bigr) \nonumber \\
& = \exp \left( \tfrac{2}{d} \, h_{\alpha}\bigl({\bf{H}}_k \, \underline{X}(n-k)\bigr) \right) \nonumber \\
\label{101}
& = \bigl| \det({\bf{H}}_k) \bigr|^{\frac{2}{d}} \, N_{\alpha}\bigl(\underline{X}(n)\bigr).
\end{align}

Let $\alpha > 1$, and $\alpha' = \frac{\alpha}{\alpha-1}$. Similarly to
Theorem~\ref{theorem: tightest REPI}, it is assumed without loss of generality
that $\bigl| \det({\bf{H}}_k) \bigr| \leq \bigl| \det({\bf{H}}_{L-1})\bigr|$
for all $k \in \{0, \ldots, L-2\}$; otherwise, the indices of
${\bf{H}}_0, \ldots, {\bf{H}}_{L-1}$ can be permuted without affecting the
differential R\'{e}nyi entropy of $\underline{Y}(n)$.
In the setting of the
improved R-EPI of Theorem~\ref{theorem: tightest REPI},
in view of \eqref{eq: c_k} and
\eqref{101}, for every $k \in \{0, \ldots, L-2\}$,
\begin{align} \label{103}
c_k = \left( \frac{\bigl| \det({\bf{H}}_k) \bigr|}{\bigl| \det({\bf{H}}_{L-1}) \bigr|}
\right)^{\frac{2}{d}}
\end{align}
which, in view of the above assumption, implies that $c_k \in [0, 1]$ for
$k \in \{0, \ldots, L-2\}$. Given the $L$ matrices $\{{\bf{H}}_k\}_{k=0}^{L-1}$,
the vector $(t_0, \ldots, t_{l-1}) \in [0,1]^L$ is calculated according to
Theorem~\ref{theorem: tightest REPI}; first $t_{L-1} \in [0,1]$ is numerically
calculated by solving the equation in \eqref{eq for tn} (with a replacement of
$1$ and $n$ in \eqref{eq for tn} by $0$ and $L-1$, respectively), and then the
rest of the $t_k$'s for $k \in \{0, \ldots, L-2\}$ are being calculated via
\eqref{eq: the t_k solution} and \eqref{eq: psi_k}. In view of
\eqref{100}, \eqref{101}, and the R-EPI of Theorem~\ref{theorem: tightest REPI},
it follows that for every~$n$
\begin{align} \label{eq: example - Th. 2}
\begin{split}
& h_{\alpha}\bigl(\underline{Y}(n)\bigr) - h_{\alpha}\bigl(\underline{X}(n)\bigr) \\
& \geq \frac{d}{2} \left( \frac{\log \alpha}{\alpha-1} + \sum_{k=0}^{L-1} g(t_k) \right)
+ \sum_{k=0}^{L-1} t_k \, \log \bigl| \det({\bf{H}}_k) \bigr|
\end{split}
\end{align}
where the function $g$ is given in \eqref{eq: g}.

In view of the derivation so far, it is easy to verify that the R-EPI in Theorem~\ref{theorem: REPI1}
is equivalent to the following looser bound, which is expressed in closed form:
\begin{align} \label{eq: example - Th. 1}
\begin{split}
& h_{\alpha}\bigl(\underline{Y}(n)\bigr) - h_{\alpha}\bigl(\underline{X}(n)\bigr) \\
& \geq \frac{d}{2} \cdot \log \left( \sum_{k=0}^{L-1} \bigl| \det({\bf{H}}_k) \bigr|^{\frac{2}{d}} \right) \\
& \hspace*{0.2cm} + \frac{d}{2} \left( \frac{\log \alpha}{\alpha-1} + \left( \frac{L \alpha}{\alpha-1} - 1 \right)
\, \log\left(1-\frac{\alpha-1}{L\alpha}\right) \right).
\end{split}
\end{align}

The R-EPI of \cite[Theorem~I.1]{BobkovC15} leads to the following loosened bound
in comparison to \eqref{eq: example - Th. 1}:
\begin{align} \label{eq: example - BC}
\begin{split}
h_{\alpha}\bigl(\underline{Y}(n)\bigr) - h_{\alpha}\bigl(\underline{X}(n)\bigr)
\geq \frac{d}{2} \left[ \log \left( \sum_{k=0}^{L-1} \bigl| \det({\bf{H}}_k) \bigr|^{\frac{2}{d}} \right)
+ \frac{\log \alpha}{\alpha-1} - \log e \right]
\end{split}
\end{align}
and, finally, the BV bound in \cite{BercherVignat} (see \eqref{BV bound - general n})
leads to the following loosening of \eqref{eq: example - Th. 2}:
\begin{align} \label{eq: example - BV}
\begin{split}
h_{\alpha}\bigl(\underline{Y}(n)\bigr) - h_{\alpha}\bigl(\underline{X}(n)\bigr)
\geq \log \left( \max_{0 \leq k \leq L-1} \bigl| \det({\bf{H}}_k) \bigr| \right).
\end{split}
\end{align}

The differential R\'{e}nyi entropy of order $\alpha \in (0,1) \cup (1, \infty)$ for a
$d$-dimensional multivariate Gaussian distribution is given by
\begin{align} \label{eq1: Gaussian}
h_{\alpha}\bigl(\underline{X}(n)\bigr) = \frac{d \log \alpha}{2(\alpha-1)}
+ \tfrac12 \, \log \Bigl( (2\pi)^d \, \det\bigl( \text{Cov}(\underline{X}(n)) \bigr) \Bigr)
\end{align}
Hence, if the entries of the Gaussian random vector $\underline{X}(n)$ are i.i.d.
\begin{align} \label{eq2: Gaussian}
h_{\alpha}\bigl(\underline{Y}(n)\bigr) - h_{\alpha}\bigl(\underline{X}(n)\bigr) = \tfrac12 \,
\log \left( \det \left( \sum_{k=0}^{L-1} {\bf{H}}_k \, {\bf{H}}_k^T \right) \right).
\end{align}

\begin{example}
Let
\begin{align}
Y(n) = 2 X(n) - X(n-1) - X(n-2)
\end{align}
for every $n$ where $\{X(n)\}$ are i.i.d. random variables, and consider the
difference $h_2(Y) - h_2(X)$ in the quadratic differential R\'{e}nyi entropy.
In this example $\alpha=2$, $d=1$, $L=3$, and $H_0=2$, $H_1=-1$, $H_2=-1$.
The lower bounds in \eqref{eq: example - Th. 2}, \eqref{eq: example - Th. 1},
\eqref{eq: example - BC}, \eqref{eq: example - BV} are equal to 0.8195,
0.7866, 0.7425 and 0.6931 nats, respectively (recall that the first two lower
bounds correspond to Theorems~\ref{theorem: tightest REPI} and~\ref{theorem: REPI1}
respectively, and the last two bounds correspond to \cite{BobkovC15} and
\cite{BercherVignat} respectively. These lower bounds are compared
to the achievable value in \eqref{eq2: Gaussian}, for an i.i.d. Gaussian input,
which is equal to 0.8959 nats.
\end{example}

\section{Summary}
\label{section: summary}
This work provides two forms of improved R\'{e}nyi entropy power inequalities
(R-EPI) for a sum of $n$ independent and continuous random vectors over $\Reals^d$.
These inequalities are of the form \eqref{Intro: R-EPI}, they refer to orders
$\alpha \in (1, \infty]$, and they also coincide with the EPI \cite{Shannon}
by letting $\alpha \to 1$. Theorem~\ref{theorem: REPI1} provides an R-EPI with a constant
which is given in closed form in \eqref{eq: c for R-EPI1}, improving the R-RPI by Bobkov and
Chistyakov in \cite[Theorem~1]{BobkovC15}; furthermore, for $n=2$, the R-EPI in
Theorem~\ref{theorem: REPI1} is asymptotically tight when $\alpha \to \infty$. The R-EPI
which is introduced in Theorem~\ref{theorem: tightest REPI} can be efficiently calculated
via a simple numerical algorithm, it is tighter than Theorem~\ref{theorem: REPI1}
and all previously reported bounds, and it is currently the best known R-EPI for
$\alpha \in (1, \infty)$. Corollary~\ref{proposition: improved bound n=2} provides
a closed-form expression for the R-EPI in Theorem~\ref{theorem: tightest REPI} for
a sum of two independent random vectors. It should be noted that the R-EPIs in
Theorems~\ref{theorem: REPI1} and~\ref{theorem: tightest REPI} coincide when the
R\'{e}nyi entropy powers of the $n$ independent random vectors are all equal.

Theorem~\ref{theorem: REPI1} is obtained by tightening the recent R-EPI by Bobkov
and Chistyakov \cite{BobkovC15} with the same analytical tools, namely the monotonicity
of $N_{\alpha}(X)$ in $\alpha$, and the use of the sharpened Young's inequality.
Theorem~\ref{theorem: tightest REPI}, which improves the tightness of the R-EPI in
Theorem~\ref{theorem: REPI1}, relies on the following additional analytical tools:
1)~a strong Lagrange duality of an optimization problem is asserted by invoking a
theorem in matrix theory \cite{Bunch} regarding the rank-one modification of a
real-valued diagonal matrix, and 2)~a solution of the Karush-Kuhn-Tucker (KKT)
equations of the related optimization problem.

\subsection*{Acknowledgement}
This work has been supported by the Israeli Science Foundation
(ISF) under Grant 12/12. A discussion with Sergio Verd\'{u} is
acknowledged. We would like to thank the Associate editor, and
the anonymous reviewers for their valuable feedback which helped
to improve the presentation in this paper.

\appendices

\section{Proof of \eqref{eq: R-EPI BC15}}
\label{appendix: BC15}
Since $\{X_k\}_{k=1}^n$ are independent, the density of $S_n=\sum\limits_{k=1}^n X_k$ is the
convolution of the densities $f_{X_k}$. In view of \eqref{eq: Young's ineq. 3} and
\eqref{eq:Pre - Renyi entropy power Norm dfn}, for $\alpha>1$,
\begin{align} \label{eq:using young}
\begin{array}{rl}
N_{\alpha}(S_n) &= \left( \|f_{X_1} \ast \ldots \ast f_{X_n}\|_\alpha \right)^{-\frac{2\alpha'}{d}} \\[0.2 cm]
                &\geq A^{-\frac{2\alpha'}{d}}\prod_{k=1}^n \left( \|f_{X_k}\|_{\nu_k} \right)^{-\frac{2\alpha'}{d}}
\end{array}
\end{align}
where
\begin{align}
\label{eq:nu_k 1}
& \nu_k > 1, \quad 1 \leq k \leq n \\
\label{eq:nu_k 2}
& \sum_{k=1}^n\frac1{\nu_k '} = \frac1{\alpha'}
\end{align}
and, due to \eqref{eq: property A} and \eqref{eq:Pre - Multiple Beckner's Best Constant},
\begin{align}
A = \left( A_{\alpha'} \prod_{k=1}^n A_{\nu_k} \right)^{\frac{d}{2}}.
\end{align}
From \eqref{eq:nu_k 1} and \eqref{eq:nu_k 2} it follows that
$\nu_k \in (1,\alpha]$ for all $k \in \{1,\ldots,n\}$, hence in view of
Corollary~\ref{corollary:Pre - Holder for densities},
\begin{align} \label{eq:using holder}
\|f_{X_k}\|_{\nu_k}^{\nu_k'} \leq
\|f_{X_k}\|_{\alpha}^{\alpha'}, \quad 1 \leq k \leq n.
\end{align}
Combining \eqref{eq:using young} and \eqref{eq:using holder}, and defining
$t_k=\frac{\alpha'}{\nu_k'}$ yields
\begin{align} \label{inequality:New REPI - Entropy Power Prod }
N_{\alpha}(S_n) \geq A^{-\frac{2\alpha'}{d}} \prod_{k=1}^n \left( \|f_{X_k}\|_{\alpha}
\right)^{-\frac{2\alpha'}{d} \cdot \frac{\alpha'}{\nu_k'}} = A^{-\frac{2\alpha'}{d}}
\prod_{k=1}^nN_{\alpha}^{t_k}(X_k)
\end{align}
which by setting $B=A^{-\frac{2\alpha'}{d}}$ completes the proof of \eqref{eq: R-EPI BC15}
with the constant $B$ as given in \eqref{eq1: BC1}.

\section{Proof of Proposition~\ref{proposition: concavity of f}}
\label{appendix: concavity of f}

In view of \eqref{eq: f function}, if $f_0$ is concave, so is $f$.
As it is verified in Section~\ref{subsection: optimization problem with f_0},
the function $f_0$ is concave for all $\alpha \in (1, 2)$ (i.e.,
$\alpha' \in (2, \infty)$), and hence also $f$ is concave for these values
of $\alpha$. We therefore need to prove the concavity of $f$ in \eqref{eq: f function}
whenever $\alpha' \in (1,2)$ (i.e., if $\alpha \in (2, \infty)$), although
$f_0$ is not concave for these values of $\alpha$.

Let $\alpha' \in (1,2)$.
If there exists an index $k \in \{1, \ldots, n-1\}$ such that $q(t_k)=0$, then
$t_k = \tfrac{\alpha'}{2} > \tfrac12$ (see \eqref{eq: q function}). In view of
\eqref{eq: polyhedron}, it follows that $t_l < \tfrac12$ for
every other index $l \neq k$ in the set $\{1, \ldots, n-1\}$, which in turn implies
from \eqref{eq: q function} that $q(t_l) < 0$ for every such index $l$. In other
words, if there exists an index $k \in \{1, \ldots, n-1\}$ such that
$q(t_k)=0$, then it follows that $q(t_l) \leq 0$ for all $l \in \{1, \ldots, n-1\}$.
In view of \eqref{eq: D and rho}, $D \preceq 0$ and $\rho < 0$ (to verify that
$\rho < 0$, note that since
$0 \leq 1- \sum_{j=1}^{n-1} t_j \leq 1-t_k=1-\tfrac{\alpha'}{2} < \tfrac12 < \tfrac{\alpha'}{2}$
then it follows from \eqref{eq: q function} and \eqref{eq: D and rho} that
$\rho = q\bigl(1-\sum_{j=1}^{n-1} t_j\bigr) < 0$); hence,
\eqref{eq1: Hessian of f} implies that $\nabla^2 f(t_1, \ldots, t_{n-1}) \prec 0$
in the interior of $\mathcal{D}^{n-1}$, so $f$ is (strictly) concave on $\mathcal{D}^{n-1}$.

To proceed, the following lemmas will be useful.
\begin{lemma} \label{lemma1: q}
If $\alpha' \in (1,2)$ and $x \in (0,1-\frac{\alpha'}{2})$, then
\begin{align}
\label{eq: lemma1 of q}
\frac1{q(x)}+\frac1{q(1-x)} > 0.
\end{align}
\end{lemma}
\begin{proof}
In view of \eqref{eq: q function}, the left side of \eqref{eq: lemma1 of q}
is equal to
\begin{align*}
\frac{\overbrace{(1-\alpha')}^{<0} \;
\overbrace{(2x^2-2x+\alpha')}^{>0}}{\underbrace{(2x-\alpha')}_{<0} \;
\underbrace{(2-2x-\alpha')}_{>0}} > 0.
\end{align*}
\end{proof}

\vspace*{0.1cm}
\begin{lemma} \label{lemma2: q}
If $\alpha' \in (1,2)$, $u,v>0$ and $u+v < 1-\frac{\alpha'}{2}$, then
\begin{align}
\label{eq: lemma2 of q}
\frac1{q(u)}+\frac1{q(1-u-v)}-\frac1{q(1-v)} > 0 .
\end{align}
\end{lemma}

\begin{proof}
In view of \eqref{eq: q function}, the left side of \eqref{eq: lemma2 of q}
is equal to
\begin{align*}
\frac{\overbrace{(2\alpha'u)}^{>0} \; \overbrace{(\alpha'+v-1)}^{>0} \;
\overbrace{(u+v-1)}^{<0}}{\underbrace{(2u-\alpha')}_{<0} \;
\underbrace{(2-2u-2v-\alpha')}_{>0} \; \underbrace{(2-2v-\alpha')}_{>0}} > 0.
\end{align*}
\end{proof}

\vspace*{0.1cm}
\begin{lemma} \label{lemma3: q}
If $n \geq 2$, $\alpha' \in (1,2)$ and
\begin{align}
\label{eq1: lemma3 of q}
\begin{split}
& t_1, \ldots, t_{n-1}>0, \\[-0.1cm]
& \sum_{k=1}^{n-1} t_k <1-\frac{\alpha'}{2}, \\[-0.1cm]
& t_n=1-\sum_{k=1}^{n-1}t_k
\end{split}
\end{align}
 then
\begin{align}
\label{eq2: lemma3 of q}
\sum_{k=1}^{n} \frac1{q(t_k)} > 0.
\end{align}
\end{lemma}

\begin{proof}
Lemma~\ref{lemma3: q} is proved by using mathematical induction on $n$.
In view of Lemma~\ref{lemma1: q}, \eqref{eq2: lemma3 of q} holds for $n=2$.
Assuming its correctness for $n$, we have
\begin{align}
\label{eq: assumption of induction}
\sum_{j=1}^{n-1} \frac1{q(t_j)} + \frac1{q(\overline{t}_n)} > 0
\end{align}
where, from \eqref{eq1: lemma3 of q},
$\overline{t}_n = 1 - \sum_{k=1}^{n-1} t_k$.
We prove in the following that \eqref{eq2: lemma3 of q} also
holds for $n+1$ when the constraints in \eqref{eq1: lemma3 of q}
are satisfied with $n+1$, i.e.,
\begin{align}
\label{eq3: lemma3 of q}
\begin{split}
& t_1, \ldots, t_n > 0, \\[-0.1cm]
& \sum_{k=1}^n t_k <1-\frac{\alpha'}{2}, \\[-0.1cm]
& t_{n+1} = 1-\sum_{k=1}^n t_k.
\end{split}
\end{align}
Consequently, the left side of \eqref{eq2: lemma3 of q} is equal to
\begin{align}
\sum_{k=1}^{n+1} \frac1{q(t_k)}
& = \sum_{k=1}^{n-1} \frac1{q(t_k)} + \frac1{q(t_n)} + \frac1{q(t_{n+1})} \nonumber \\
\label{eq: by induction}
& > -\frac1{q(\overline{t}_n)} + \frac1{q(t_n)} + \frac1{q(t_{n+1})} \\
\label{eq: by equality constraint}
& = \frac1{q(t_n)} + \frac1{q\left(1-\sum_{k=1}^n t_k \right)} - \frac1{q(1-\sum_{k=1}^{n-1} t_k)} \\
\label{eq: by lemma 2}
& > 0
\end{align}
where \eqref{eq: by induction} follows from \eqref{eq: assumption of induction};
\eqref{eq: by equality constraint} holds by the equality constraint in
\eqref{eq3: lemma3 of q}; \eqref{eq: by lemma 2} follows from Lemma~\ref{lemma2: q}
by setting $u = t_n$, $v = \sum_{k=1}^{n-1} t_k$ which satisfy
$u+v < 1-\frac{\alpha'}{2}$ in view of \eqref{eq3: lemma3 of q}.
Hence, it follows by mathematical induction that Lemma~\ref{lemma3: q}
holds for every $n \geq 2$.
\end{proof}

In the following, we prove the concavity of $f$ when $q(t_k) \neq 0$
for all $k \in \{1, \ldots, n-1\}$ (recall that the case where there
exits $k \in \{1, \ldots, n-1\}$ such that $q(t_k)=0$ was addressed
in the paragraph before Lemma~\ref{lemma1: q}). Without loss of
generality, we prove that $\nabla^2 f(\underline{t}) \preceq 0$ when
$\bigl(q(t_1), \ldots, q(t_{n-1})\bigr)$ is
a vector whose all entries are distinct. To justify this assumption,
note that since the function $q$ in \eqref{eq: q function} is
monotonically increasing ($q'(t) = \frac1{t^2} + \frac1{(\alpha'-t)^2} > 0$),
we actually restrict ourselves under the latter assumption to the
case where the entries of the vector
$(t_1, \ldots, t_{n-1})$ are all distinct. Otherwise, if some of
the entries of the vector $(t_1, \ldots, t_{n-1})$ are equal,
then the proof that the Hessian matrix is
non-positive definite continues to hold by relying on the satisfiability
of this property when all the entries of $(t_1, \ldots, t_{n-1})$ are distinct,
and from the continuity in $\underline{t}$ of the eigenvalues
of the Hessian matrix $\nabla^2 f(\underline{t})$.

Since the optimization problem in \eqref{opt. prob. with n-1 variables} is
invariant to a permutation of the entries of $\underline{t}$, it is assumed
without loss of generality that
\begin{align} \label{eq: order of q}
q(t_1) < q(t_2) < \ldots < q(t_{n-1}).
\end{align}
In view of \eqref{eq: order of q}, there are only two possibilities: either
\begin{align} \label{case 1.1}
q(t_1) < q(t_2) < \ldots < q(t_{n-2}) < q(t_{n-1}) < 0,
\end{align}
or
\begin{align} \label{case 1.2}
q(t_1) < q(t_2) < \ldots < q(t_{n-2}) < 0 < q(t_{n-1})
\end{align}
as if it was possible that $q(t_{n-2}) \geq 0$, it would have implied
that $q(t_{n-1}) > q(t_{n-2}) \geq 0$ which in turn yields that
$t_{n-1} > t_{n-2} \geq \frac{\alpha'}{2}$. This, however, cannot be
true since otherwise
$$\sum_{k=1}^{n-1} t_k \geq t_{n-2} + t_{n-1} > \alpha' > 1$$
which violates the inequality constraint $\sum_{k=1}^{n-1} t_k \leq 1$
in \eqref{eq: polyhedron}.

\vspace*{0.1cm}
The continuation of this proof relies on Fact~\ref{fact: Bunch} by Bunch
{\em et al.} \cite{Bunch} (see Section~\ref{subsection:Pre - Rank-One Modification}),
and on Lemma~\ref{lemma3: q}. For the continuation of this proof, let
\begin{align} \label{eq: last t}
t_n = 1 - \sum_{k=1}^{n-1} t_k.
\end{align}

\vspace*{0.1cm}
{\em Case~1}:
If \eqref{case 1.1} holds, then \eqref{eq: D and rho} implies that
\begin{align} \label{eq: D is negative definite}
D \prec 0.
\end{align}
\begin{itemize}
\item If $q(t_n) < 0$ then $\rho = q(t_n) {\bf{1}} {\bf{1}}^T  \prec 0$ which,
in view of \eqref{eq1: Hessian of f} and \eqref{eq: D is negative definite},
implies that $\nabla^2 f(t_1, \ldots, t_{n-1}) \prec 0$.
\item Otherwise, if $q(t_n)>0$ then $\rho > 0$ (see \eqref{eq: D and rho} and \eqref{eq: last t});
from \eqref{eq1: Hessian of f} and the interlacing property in
\eqref{eq: Bunch1}, the eigenvalues $\lambda_1, \ldots, \lambda_{n-1}$ of
$\nabla^2 f(\underline{t})$ satisfy
\begin{align} \label{eq1: ineq. for the lambda's}
q(t_1) < \lambda_1 < q(t_2) < \ldots < q(t_{n-2}) < \lambda_{n-2} < q(t_{n-1}) < \lambda_{n-1}
\end{align}
where, in view of the third item of Fact~\ref{fact: Bunch},
the inequalities in \eqref{eq1: ineq. for the lambda's} are strict. From
\eqref{case 1.1} and \eqref{eq1: ineq. for the lambda's}, it follows that
$\lambda_1, \ldots, \lambda_{n-2} < 0$. To prove that
$\nabla^2 f(t_1, \ldots, t_{n-1}) \prec 0$, it remains to show that also
$\lambda_{n-1} < 0$.
In view of the third item of Fact~\ref{fact: Bunch} and
\eqref{eq1: Hessian of f}, the eigenvalues of
$\nabla^2 f(t_1, \ldots, t_{n-1})$ satisfy the equation
\begin{align}
1 + q(t_n) \, \sum_{j=1}^{n-1} \frac1{q(t_j)-\lambda} = 0
\end{align}
which therefore implies that, for all $k \in \{1, \ldots, n-1\}$,
\begin{align} \label{eq: from Bunch}
\sum_{j=1}^{n-1} \frac1{\lambda_k - q(t_j)} = \frac1{q(t_n)}.
\end{align}
Let us assume on the contrary that $\lambda_{n-1} > 0$.
Since it is assumed here that $q(t_n)>0$ then $t_n > \frac{\alpha'}{2}$, and
it follows from \eqref{eq: last t} that
\begin{align} \label{1}
\sum_{k=1}^{n-1} t_k < 1- \frac{\alpha'}{2}.
\end{align}
Since $q(t_j) < 0$ for all $j \in \{1, \ldots, n-1\}$, if $\lambda_{n-1} > 0$, then
in view of \eqref{eq: from Bunch}
\begin{align}
\begin{split} \label{2}
\sum_{j=1}^{n-1} \frac1{-q(t_j)} & \geq \sum_{j=1}^{n-1}
\frac1{\lambda_{n-1} - q(t_j)} \\
& = \frac1{q(t_n)}.
\end{split}
\end{align}
Rearrangement of terms in \eqref{2} yields
\begin{align} \label{3}
\sum_{j=1}^n \frac1{q(t_j)} \leq 0
\end{align}
and, in view of the interior of $\mathcal{D}^{n-1}$ in \eqref{eq: polyhedron},
and \eqref{eq: last t} and \eqref{1},
inequality \eqref{3} contradicts the result in Lemma~\ref{lemma3: q}. This
therefore proves by contradiction that $\lambda_{n-1} < 0$, so all the
$n-1$ eigenvalues of the Hessian are negative, and therefore $f$ is strictly
concave under the assumption in \eqref{case 1.1}.
\end{itemize}

\vspace*{0.2cm}
{\em Case 2}: We now consider the case where \eqref{case 1.2} holds.
Under this assumption,
\begin{align}
\label{eq: q<0}
q(t_n) < 0.
\end{align}
To verify \eqref{eq: q<0},
note that $q(t_{n-1}) > 0$ yields that $t_{n-1} > \frac{\alpha'}{2}$;
assume by contradiction that $q(t_n) \geq 0$, then $t_n \geq \frac{\alpha'}{2}$
(see \eqref{eq: q function}) which implies that
$\sum_{j=1}^n t_j \geq t_n + t_{n-1} > \alpha' > 1$ in contradiction to
the equality $\sum_{j=1}^n t_j=1$ in \eqref{eq: last t}; hence, indeed
$q(t_n) < 0$. Consequently, in view of \eqref{eq1: Hessian of f}, let
\begin{align} \label{eq1: bar C}
\overbar{C} = & \frac1{q(t_n)} \; \nabla^2f(t_1, \ldots, t_{n-1}) \\
\label{eq2: bar C}
&= \overbar{D} + {\bf{1}} {\bf{1}}^T
\end{align}
where
\begin{align} \label{matrix D}
\overbar{D} = \text{diag}\left( \frac{q(t_1)}{q(t_n)}, \ldots,
\frac{q(t_{n-1})}{q(t_n)} \right).
\end{align}
From \eqref{case 1.2} and \eqref{eq: q<0}, it follows that
\begin{align}
\frac{q(t_1)}{q(t_n)} > \frac{q(t_2)}{q(t_n)} > \ldots >
\frac{q(t_{n-2})}{q(t_n)} > 0 > \frac{q(t_{n-1})}{q(t_n)}.
\end{align}
It is shown in the following that $\overbar{C} \succeq 0$ which,
from \eqref{eq: q<0} and \eqref{eq1: bar C}, imply that indeed
$\nabla^2 f(t_1, \ldots, t_{n-1}) \preceq 0$. Let $\{\lambda_k\}_{k=1}^{n-1}$
designate the eigenvalues of $\overbar{C}$; in view of \eqref{eq2: bar C}
and the last two items of Fact~\ref{fact: Bunch}, it follows that
\begin{align} \label{4}
\overbrace{\frac{q(t_{n-1})}{q(t_n)}}^{<0} < \lambda_1
< \overbrace{\frac{q(t_{n-2})}{q(t_n)}}^{>0}
< \lambda_2 < \ldots < \overbrace{\frac{q(t_2)}{q(t_n)}}^{>0}
< \lambda_{n-2} < \overbrace{\frac{q(t_1)}{q(t_n)}}^{>0} < \lambda_{n-1}.
\end{align}
Hence, \eqref{4} asserts that $\lambda_2, \ldots, \lambda_{n-1} > 0$, and it
only remains to prove that $\lambda_1 > 0$.
From the third item of Fact~\ref{fact: Bunch}, and from \eqref{eq1: bar C},
\eqref{eq2: bar C}, \eqref{matrix D}, the eigenvalues
$\{\lambda_k\}_{k=1}^n$ of the rank-one modification $\overbar{C}$ satisfy
the equality
\begin{align} \label{5}
1 + \sum_{j=1}^{n-1} \frac1{\frac{q(t_j)}{q(t_n)} - \lambda_k} = 0
\end{align}
for all $k \in \{1, \ldots, n-1\}$.
Assume on the contrary that $\lambda_1 \leq 0$, then from \eqref{5}
\begin{align}
& 1 + \sum_{j=1}^{n-1} \frac{q(t_n)}{q(t_j)}
\geq 1 + \sum_{j=1}^{n-1} \frac1{\frac{q(t_j)}{q(t_n)} - \lambda_1} = 0. \label{6}
\end{align}
Consequently, from \eqref{eq: q<0} and \eqref{6}, it follows that
$\sum_{j=1}^n \frac1{q(t_j)} \leq 0$ in contradiction to
Lemma~\ref{lemma3: q}. Hence, all $\lambda_k > 0$ for
$k \in \{1, \ldots, n-1\}$, which therefore implies that
$\nabla^2 f(t_1, \ldots, t_{n-1}) \prec 0$ for all $(t_1, \ldots, t_{n-1})$
in the interior of $\mathcal{D}^{n-1}$.
This completes the proof of Proposition~\ref{proposition: concavity of f}.

\section{Derivation of \eqref{eq1s: KKT}--\eqref{eq3s: KKT} From Lagrange Duality}
\label{appendix: Lagrangian}

We consider the convex optimization problem in \eqref{opt. prob. with n-1 variables},
and solve it via the use of the Lagrange duality where strong duality holds.

The Lagrangian of the convex optimization problem in \eqref{opt. prob. with n-1 variables} is given by
\begin{align} \label{eq: Lagrangian}
\begin{split}
& L(t_1,\ldots,t_{n-1};\lambda_1,\ldots,\lambda_n) \\
&= \sum_{k=1}^{n-1} g(t_k) + g\left(1 - \sum_{k=1}^{n-1} t_k \right)
+ \sum_{k=1}^{n-1} t_k \log N_k \\
& \hspace*{0.3cm} + \left(1-\sum_{k=1}^{n-1}t_k \right)\log N_n
+ \sum_{k=1}^{n-1} \lambda_k t_k + \lambda_n \left(1-\sum_{k=1}^{n-1} t_k \right)
\end{split}
\end{align}
where $\lambda \succeq 0$, the function $g$ is defined in \eqref{eq: g}, and $N_k := N_{\alpha}(X_k)$
(see \eqref{eq: N_alpha entries}).

In view of the Lagrangian in \eqref{eq: Lagrangian} and the function $g$ defined in
\eqref{eq: g}, straightforward calculations of
the partial derivatives of $L$ with respect to $t_k$ for $k \in \{1, \ldots, n-1\}$
yields
\begin{align}
\frac{\partial L}{\partial t_k} &= g'(t_k) - g'(1-t_1 - \ldots -t_{n-1}) +
\log \left(\frac{N_{\alpha}(X_k)}{N_{\alpha}(X_n)}\right) + \lambda_k - \lambda_n \nonumber \\
&= -\log \left( t_k \Bigl(1-\frac{t_k}{\alpha'}\Bigr) \right) +
\log \left( t_n \Bigl(1-\frac{t_n}{\alpha'}\Bigr) \right) +
\log\left(\frac{N_{\alpha}(X_k)}{N_{\alpha}(X_n)}\right)+\lambda_k-\lambda_n
\label{eq: partial derivative of L}
\end{align}
where $t_n := 1 - \sum_{k=1}^{n-1} t_k$. By setting the partial derivatives in
\eqref{eq: partial derivative of L} to zero, and exponentiating both sides of the
equation, we get for all $k \in \{1, \ldots, n-1\}$
\begin{align}
\frac{t_n (\alpha' - t_n)}{t_k (\alpha' - t_k)} = \frac{N_{\alpha}(X_n)}{N_{\alpha}(X_k)}
\cdot \exp(\lambda_n - \lambda_k). \label{eq1a: KKT}
\end{align}
In view of \eqref{eq1a: KKT} and the definition of $\{c_k\}_{k=1}^{n-1}$ in \eqref{eq: c_k},
we obtain that for all $k \in \{1, \ldots, n-1\}$
\begin{align}
t_k (\alpha' - t_k) = c_k \, t_n (\alpha' - t_n) \, \exp(\lambda_k - \lambda_n). \label{eq1b: KKT}
\end{align}
Consequently, \eqref{eq1b: KKT}, the definition of $t_n$, and the slackness conditions lead to the
following set of constraints:
\begin{align}
\label{eq1: KKT}
& t_k  \geq 0, \quad k \in \{1, \ldots, n\} \\
\label{eq2: KKT}
& \sum_{k=1}^n t_k = 1 \\
\label{eq3: KKT}
& \lambda_k \geq 0, \quad k \in \{1, \ldots, n\} \\
\label{eq4: KKT}
& \lambda_k t_k = 0, \qquad k \in \{1, \ldots, n\} \\
\label{eq5: KKT}
& t_k (\alpha'-t_k)= c_k t_n (\alpha'-t_n) \, \exp(\lambda_k-\lambda_n),
\quad k \in \{1, \ldots, n-1\}
\end{align}
with the variables $\underline{\lambda}$ and $\underline{t}$ in
\eqref{eq1: KKT}--\eqref{eq5: KKT}.

Consider first the case where
\begin{align} \label{eq:N_k>0}
N_\alpha(X_k)>0,\quad \forall \, k \in \{1,\ldots,n-1\}
\end{align}
which in view of \eqref{eq: c_k}, implies
\begin{align} \label{eq: c_k>0}
c_k >0,\quad \forall \, k \in \{1,\ldots,n-1\}.
\end{align}
Under the assumption in \eqref{eq:N_k>0}, we prove that
\begin{align} \label{all lambda's are zero}
\lambda_k=0, \quad \forall \, k \in \{1, \ldots, n\}.
\end{align}
Assume on the contrary that there exists an index $k$ such that
$\lambda_k \neq 0$. This would imply from \eqref{eq4: KKT} that
$t_k = 0$. If $k=n$ (i.e., if $t_n=0$) then it follows from
\eqref{eq5: KKT} that also $t_k=0$ for all $k \in \{1, \ldots, n\}$
(recall that $\alpha' > 1$),
which violates the equality constraint in \eqref{eq2: KKT}. Otherwise,
if $t_k = 0$ for some $k < n$, then it follows from \eqref{eq5: KKT}
and \eqref{eq: c_k>0} that $t_n=0$ which leads to the same contradiction as above.

The substitution of \eqref{all lambda's are zero} into the right side of
\eqref{eq5: KKT} gives the simplified equation in \eqref{eq1s: KKT}. In
view of \eqref{eq1: KKT} and \eqref{eq2: KKT}, this leads to the simplified
set of KKT constraints in \eqref{eq1s: KKT}--\eqref{eq3s: KKT}.

Finally, if the assumption in \eqref{eq:N_k>0} does not hold, i.e.,
$N_\alpha(X_k)=0$ for some $k \in \{1,\ldots,n-1\}$, then
the optimal solution satisfies $t_k=0$ (with the convention
that $0\cdot \log 0 = 0$) since any other assignment makes
the objective function in \eqref{eq: rewriting f} be equal to $-\infty$.
In addition, in this case $c_k=0$, so the simplified set of KKT
constraints in \eqref{eq1s: KKT}--\eqref{eq3s: KKT} still yields
the optimal solution $\underline{t}$.

\section{On the existence and uniqueness of the solution to \eqref{eq for tn}}
Define
\begin{align} \label{eq:phi}
\phi_{\alpha}(x)=x+\sum_{k=1}^{n-1}\psi_{\alpha,k}(x),\quad x \in [0,1],
\end{align}
and note that we need to show that there exists a unique solution of the equation
$\phi_\alpha(x)=1$ where $x \in [0,1]$. From the continuity of $\phi_\alpha(\cdot)$
and since $\phi_\alpha(0) = 0 $ and
\begin{align} \label{eq: phi at 1 > 1}
\phi_\alpha(1) =1+ \sum_{k=1}^{n-1}\psi_k(1) > 1,
\end{align}
the existence of such a solution is assured. To prove uniqueness, consider two cases:
$\alpha' \geq 2$ and $1 < \alpha' < 2$.

The derivative of $\phi_\alpha(x)$ is given by
\begin{align} \label{eq:phi'}
\phi_{\alpha}'(x) = 1 + \sum_{k=1}^{n-1} \frac{c_k(\alpha'-2x)}{\sqrt{\alpha'^2-4c_k
x(\alpha'-x)}},
\end{align}
so if $\alpha' \geq 2$, then $\phi_\alpha(x)$ is monotonically increasing in $[0,1]$,
hence the solution $t_n \in [0,1]$ of the equation \eqref{eq for tn} is unique.

If $\alpha' \in (1,2)$, then
\begin{align} \label{eq:phi mono}
\phi_\alpha'(x)>0, \qquad x \in [0,\tfrac{\alpha'}{2}].
\end{align}
Note that
\begin{align*}
\alpha'^{\, 2}-4c_kx(\alpha'-x) =\alpha'^{\, 2}(1-c_k)+c_k(2x-\alpha')^2,
\end{align*}
thus in view of \eqref{eq:phi'},
\begin{align} \label{eq:phi' 2}
\phi_\alpha'(x) =1 + \sum_{k=1}^{n-1} \frac{ c_k} {\sqrt{c_k+
 \frac{\alpha'^2(1-c_k)} {4(x-\frac{\alpha'}{2})^2}}}.
\end{align}
Eq.~\eqref{eq:phi' 2} implies that $\phi_\alpha'(\cdot)$ is monotonically decreasing in
$(\frac{\alpha'}{2}, 1]$; in other words, $\phi_\alpha(\cdot)$ is concave in the interval
$(\tfrac{\alpha'}{2}, 1]$).

Assume on the contrary that there are two solutions, $0<x_1<x_2<1$ to \eqref{eq for tn}, i.e.,
\begin{align} \label{eq:assumption}
\phi_\alpha(x_1)=\phi_\alpha(x_2)=1.
\end{align}
Eq.~\eqref{eq:assumption} implies that there exists $c \in (x_1, x_2)$
such that $\phi_{\alpha}'(c)=0$ and from \eqref{eq:phi mono}, $c \in (\frac{\alpha'}{2}, x_2)$.
Since $\phi_\alpha'(\cdot)$ is monotonically decreasing in $(\frac{\alpha'}{2}, 1]$, it follows that
$\phi_{\alpha}'(x)<0$ for all $x \in (c,1)$. Hence, $\phi_\alpha(\cdot)$ is monotonically
decreasing in $(x_2,1)$, which leads to the contradiction $$1<\phi_\alpha(1)<\phi_\alpha(x_2)=1.$$
This therefore demonstrates the uniqueness of the solution in both cases.

\label{appendix: Exist and Uniq}

\section{On the Asymptotic Equivalence of \eqref{eq: tightest REPI} and \eqref{BV bound - general n}}
\label{appendix: Tightest REPI meets BV}
If $N_\infty(X_k)=0$ for all $k \in \{1,\ldots,n\}$, the bounds in \eqref{eq: tightest REPI}
and \eqref{BV bound - general n} obviously coincide asymptotically as $\alpha \to \infty$.
In addition, in this case, the condition in \eqref{Tightest REPI meets BV} clearly holds
as well. It is therefore assumed that $N_\infty(X_k)$ is strictly positive for at least one
value of $k \in \{1,\ldots,n\}$ which, under the assumption in \eqref{eq: last entry of N is maximal},
yields that
\begin{align} \label{eq: max N_infty not 0}
 N_\infty(X_n)> 0.
\end{align}

Let $c_k^\star$ be defined as
\begin{align} \label{eq: c_k^infty }
c_k^\star=\lim\limits_{\alpha \to \infty}
\frac{N_\alpha(X_k)}{N_\alpha(X_n)}=\frac{N_\infty(X_k)}{N_\infty(X_n)}.
\end{align}
In view of \eqref{eq: c_k^infty }, the condition in \eqref{Tightest REPI meets BV}
is equivalent to
\begin{align} \label{eq: c_k condition}
\sum_{k=1}^{n-1}c_k^\star \leq 1.
\end{align}
Hence, it remains to show that the the tightest R-EPI in \eqref{eq: tightest REPI} and
the BV bound in \eqref{BV bound - general n} asymptotically coincide, by letting
$\alpha \to \infty$, if and only if the condition in \eqref{eq: c_k condition} holds.

Let $\phi_\alpha \colon [0,1] \to \Reals$ be the function defined in \eqref{eq:phi}
for $\alpha \in (1, \infty)$, and define
\begin{align}
\label{eq: phi infty}
&\phi_\infty(x)=\lim\limits_{\alpha \to \infty}\phi_\alpha(x)
\end{align}
for $x \in [0,1]$.
In view of \eqref{eq: psi_k}, \eqref{eq:phi} and \eqref{eq: c_k^infty },
the limit in \eqref{eq: phi infty} is given by
\begin{align} \label{eq: phi infty equals}
\phi_\infty(x)= x+ \tfrac12 \, \sum_{k=1}^{n-1} \left(1 - \sqrt{1 - 4 c_k^\star \, x (1-x)} \, \right)
\end{align}
for $x \in [0,1]$.
Recall that under the assumption in \eqref{eq: last entry of N is maximal}, the selection of
$t_n=1$ in \eqref{eq: rewriting f} leads to the BV bound in \eqref{BV bound - general n}. Hence,
in view of \eqref{eq for tn}, if $t=1$ is the unique solution of
\begin{align} \label{eq: phi_infty is 1}
\phi_\infty(t)=1, \quad t \in [0,1]
\end{align}
then the bounds in \eqref{eq: tightest REPI} and \eqref{BV bound - general n} asymptotically coincide
by letting $\alpha \to \infty$.
Note that,
\begin{align}
\label{eq: phi_infty at 0}
\phi_\infty(0) = 0, \\[0.2 cm]
\label{eq: phi_infty at 1}
\phi_\infty(1) = 1.
\end{align}
From \eqref{eq: phi_infty at 1}, $t=1$ is a solution of \eqref{eq: phi_infty is 1} regardless of the sequence
$\{c_k^\star\}$. Moreover, from \eqref{eq: phi infty equals},
\begin{align} \label{eq: phi's derivative}
\phi_\infty'(x) =1 + \sum_{k=1}^{n-1} \frac{c_k^\star \, (1-2x)}{\sqrt{1-4c_k^\star \, x(1-x)}},
\end{align}
so
\begin{align}
\label{eq: phi's derivative 2}
&\phi_\infty'(x)>0, \quad \forall \, x \in (0,\tfrac12), \\
\label{eq: phi's derivative 3}
&\phi_\infty'(1)=1-\sum_{k=1}^{n-1} c_k^\star.
\end{align}
The function $\phi_\infty'(\cdot)$ is monotonically decreasing in the interval
$[\tfrac12,1]$; this concavity property
of $\phi_{\infty}$ can be justified by Appendix~\ref{appendix: Exist and Uniq} since
the function $\phi_{\alpha}(\cdot)$ is concave in $[\tfrac{\alpha'}{2}, 1]$ and
$\alpha' \to 1$ by letting $\alpha \to \infty$.
Thus, if the condition in \eqref{eq: c_k condition} holds, then $\phi_\infty'(x)>0$ for all
$ x \in (0,1)$ which, in view of \eqref{eq: phi_infty at 1}, yields that $t=1$ is the unique
solution of \eqref{eq: phi_infty is 1}. This implies that the tightest R-EPI in
\eqref{eq: tightest REPI} and the BV bound in \eqref{BV bound - general n} asymptotically
coincide by letting $\alpha \to \infty$.

\vspace*{0.2cm}
To prove the 'only if' part, one needs to show that if the condition in \eqref{eq: c_k condition}
does not hold then the bounds in \eqref{eq: tightest REPI} and
\eqref{BV bound - general n} do not coincide asymptotically in the limit where $\alpha \to \infty$;
in the latter case, we prove that our bound in \eqref{eq: tightest REPI} is tighter than
\eqref{BV bound - general n}. If \eqref{eq: c_k condition} does not hold, then
\eqref{eq: phi's derivative 3} implies that
\begin{align} \label{eq: phi's derivative 4}
\phi_\infty'(1) < 0.
\end{align}
Hence, from \eqref{eq: phi_infty at 1}, there exists $x_0 \in (0,1)$ such that $\phi_\infty(x_0)>1$
which, in view of \eqref{eq: phi_infty at 0} and the continuity of $\phi_\infty(\cdot)$, implies
that there exists $t \in (0,x_0)$ which is a solution of \eqref{eq: phi_infty is 1}. This implies
that there are two different solutions of \eqref{eq: phi_infty is 1} in the interval $[0,1]$.
Let $t^{(1)} \in (0,1)$ and $t^{(2)}=1$ denote such solutions, i.e.,
\begin{align} \label{t1<t2}
t^{(1)}<t^{(2)}=1.
\end{align}
Note that there are no solutions of the equation $\phi_{\infty}(t)=1$ in $[0,1]$, except for
$t^{(1)}$ and $t^{(2)}=1$ since $\phi_{\infty}(\cdot)$ is monotonically increasing in $[0, \tfrac12]$
and it is concave in $[\tfrac12, 1]$ with $\phi_{\infty}(1)=1$.

We need to show that $t^{(1)}$ leads to an R-EPI which is tighter than the R-EPI in
\eqref{BV bound - general n}; the bound in \eqref{BV bound - general n} corresponds
to $t^{(2)}=1$ under the assumption in \eqref{eq: last entry of N is maximal}. For
every $\alpha>1$, let $t(\alpha)$ be the unique solution of \eqref{eq for tn} (see
Appendix~\ref{appendix: Exist and Uniq}). It follows that the limit of any convergent subsequence
$\{t(\alpha_n)\}$, as $\alpha_n \to \infty$, is either $t^{(1)} \in (0,1)$
or $t^{(2)}=1$. In the sequel, if the condition in \eqref{Tightest REPI meets BV} is not
satisfied, we show that every such subsequence tends to $t^{(1)} \in (0,1)$, which therefore
implies that
\begin{align} \label{lim t_n alpha}
\lim\limits_{\alpha \to \infty} t(\alpha) = t^{(1)} < 1.
\end{align}
From \eqref{eq: phi's derivative 4} and the continuity of $\phi_\infty(\cdot)$,
it follows that there exists $\delta>0$ such that
\begin{align} \label{phi infty > 1}
\phi_\infty(x)>1,\; \forall \, x \in (1-\delta,1).
\end{align}
In addition, since $\phi_\alpha(\cdot)$ is continuous in $\alpha$ for every $x \in [0,1]$, it follows
from \eqref{phi infty > 1} that there exists $\alpha_0>1$ such that $ \phi_\alpha(x)>1$
for all $\alpha>\alpha_0$ and $x \in (1-\delta,1]$ (note that the rightmost point is included in this
interval in view of \eqref{eq: phi at 1 > 1}).
Hence, since by definition $\phi_{\alpha}\bigl(t(\alpha)\bigr)=1$ for all $\alpha \in (1, \infty)$
then $t(\alpha)\leq 1-\delta$ for all $\alpha>\alpha_0$. This therefore proves that every
subsequence $\{t(\alpha_n)\}$ tends to $t^{(1)}$ as $\alpha_n \to \infty$ (since it cannot
converge to $t^{(2)}=1$), which yields \eqref{lim t_n alpha}. Hence, the R-EPI in
Theorem~\ref{theorem: tightest REPI} asymptotically yields a tighter bound than \eqref{BV bound - general n}
when $\alpha \to \infty$; this therefore proves the 'only if' part of our claim.


\end{document}